\documentclass[a4paper,reqno,oneside]{amsart}
\usepackage{etoolbox} 

\newtoggle{llncs}
\makeatletter
\@ifclassloaded{llncs}{\toggletrue{llncs}}{\togglefalse{llncs}}
\makeatother
\iftoggle{llncs}{}{
 \author{Thomas Debris--Alazard$^{1}$}
\email{thomas.debris@inria.fr}
\author{Philippe Gaborit$^{2}$}
\email{gaborit@unilim.fr}
\author{Romaric Neveu$^{2}$}
\email{romaric.neveu@etu.unilim.fr}
\author{Olivier Ruatta$^{2}$}
\email{olivier.ruatta@unilim.fr}

\address{$^{1}$ Inria and  Laboratoire LIX, \'Ecole Polytechnique, Palaiseau, France}

\address{$^{2}$ XLIM, University of Limoges, Limoges, France}

 }
\iftoggle{llncs}{
	\usepackage{amsmath}
\usepackage{amssymb}
\usepackage{bbm}
\usepackage[english]{babel}
\usepackage{color}
\usepackage{enumitem}
\usepackage{ellipsis}
\usepackage[breaklinks,colorlinks=true, citecolor=blue]{hyperref}
\usepackage[misc,geometry]{ifsym} \usepackage{dsfont}
\usepackage{mathrsfs}
\usepackage{mathtools}
\usepackage{algorithm}
\usepackage{algpseudocode}
\usepackage{stmaryrd}
\usepackage{thm-restate}
\usepackage{subcaption}
\usepackage{systeme}

\usepackage{tikz}
\usetikzlibrary{patterns}
\usetikzlibrary{matrix}

\usepackage[keys]{cryptocode}
\usepackage{comment}
\usepackage{todonotes}
\usepackage{lipsum}
\usepackage{makecell, multirow}
\usepackage{pifont}
\pagestyle{plain}
\usetikzlibrary{fit}
\usepackage[most]{tcolorbox}

 }{
	\pagestyle{plain}
\usepackage[utf8]{inputenc}
\usepackage{a4wide}
\usepackage{amsmath}
\usepackage{amssymb}
\usepackage[english]{babel}
\usepackage{color}
\usepackage{enumitem}
\usepackage[breaklinks,colorlinks=true, citecolor=blue]{hyperref}
\usepackage{dsfont}
\usepackage{mathrsfs}
\usepackage{mathtools}
\usepackage{algorithm}
\usepackage[noend]{algpseudocode}
\usepackage{graphicx}
\usepackage{stmaryrd}
\usepackage{thm-restate}
\usepackage[normalem]{ulem}
\usepackage{tikz}
\usetikzlibrary{matrix}
\usetikzlibrary{arrows}
\usetikzlibrary{patterns, shapes, positioning}
\usetikzlibrary{decorations.pathmorphing,shapes}
\usepackage{makecell, multirow} 
\usepackage{setspace}
\usepackage{xspace}
\usepackage[keys]{cryptocode}
\usepackage{comment}
\usepackage{bbm}
\usepackage{pifont}
\usepackage[foot]{amsaddr}
\usepackage{graphicx} \usetikzlibrary{fit}

\usepackage[most]{tcolorbox}

 }

\newcommand{\Fq}{\mathbb{F}_q}

\newcommand{\Fqm}{\mathbb{F}_{q^m}}

\renewcommand{\vec}[1]{\mathbf{#1}}

\newcommand{\vv}{\vec{v}}

\newcommand{\minrank}{\mathsf{MinRank}}
\newcommand{\MinRank}{\mathsf{MinRank}}

\iftoggle{llncs}{
	\spnewtheorem{fact}{Fact}{\bfseries}{\itshape}}{
	\newtheorem{remark}{Remark}
	\newtheorem{lemma}{Lemma}
	
	\newtheorem{definition}{Definition}
	\newtheorem{theorem}{Theorem}
	\newtheorem{proposition}{Proposition}

	\newtheorem{corollary}{Corollary}
	
}
\newtheorem{modeling}{Modeling}
\newtheorem{heuristic}{Heuristic}
\algnewcommand{\Keygen}{\item[\textbullet \textbf{~KeyGen:}]}
\algnewcommand{\Enc}{\item[\textbullet \textbf{~Encryption:}]}
\algnewcommand{\Dec}{\item[\textbullet \textbf{~Decryption:}]}

\newcommand{\eqdef}{\stackrel{\textup{def}}{=}}

\newcommand{\rk}{\mathsf{Rank}}
\newcommand{\tr}{\mathsf{Trace}}

\newcommand{\qbinom}[2]{\begin{bmatrix}#2 \\ #1\end{bmatrix}_q}

\newcommand{\RSD}{\ensuremath{\mathsf{RSD}}}
\newcommand{\RSL}{\ensuremath{\mathsf{RSL}}}

\newcommand{\MSL}{\mathsf{MSL}}

\newcommand{\adv}{\mathcal A}

\newcommand{\C}{\mathcal C}

\newcommand{\Unif}{\hookleftarrow}

\newcommand{\oneto}[1]{[1,#1]}

\newcommand{\cmark}{\textcolor{green}{\ding{51}}}\newcommand{\xmark}{\textcolor{red}{\ding{55}}}\newcommand{\Sp}[1]{\mathbf{Sp}(#1)}

	\title{A Minrank-based Encryption Scheme \`a la Alekhnovich-Regev}

\begin{document}

\maketitle

\begin{abstract}
		
		Introduced in 2003 and 2005, Alekhnovich and Regev' schemes were the first public-key encryptions whose security is only based on the average hardness of decoding random linear codes and $\mathsf{LWE}$, without other security assumptions. Such security guarantees made them very popular, being at the origin of the now standardized \textsf{HQC} or \textsf{Kyber}.

		We present an adaptation of Alekhnovich and Regev' encryption scheme whose security is only based on the hardness of a slight variation of $\minrank$, the so-called stationary-$\minrank$ problem. We succeeded to reach this strong security guarantee by showing that stationary-$\minrank$ benefits from a search-to-decision reduction. Our scheme therefore brings a partial answer to the long-standing open question of building an encryption scheme whose security relies solely on the hardness of $\minrank$. 
Finally, we show after a thoroughly security analysis that our scheme is practical and competitive with other encryption schemes admitting such strong security guarantees. Our scheme is slightly less efficient than \textsf{FrodoKEM}, but much more efficient than Alekhnovich and Regev' original schemes, with possibilities of improvements by considering more structure, in the same way as \textsf{HQC} and \textsf{Kyber}.
\end{abstract}

\section{Introduction}\label{sec:intro}

{\bf \noindent Post-quantum encryption schemes: the case of codes and lattices.}
Among all the candidates for post-quantum cryptography, codes and lattices have proven themselves to be strong candidates. This success, culminating in the standardization of \textsf{Kyber} (now \textsf{ML-KEM}) and \textsf{HQC} as key-exchange mechanisms, is the result of a long line of work started in 2003 by Alekhnovich~\cite{A03} for codes and in 2005 by Regev~\cite{R05} for lattices. In fact, while code-based encryptions
existed since 1978 with McEliece's scheme \cite{ME78},
Alekhnovich and Regev' schemes showed a core difference concerning their security hypothesis. This was the first time that security \textit{only} relied on the hardness of decoding a random linear code and $\textsf{LWE}$, both problems benefiting from many sources of hardness like search-to-decision reductions~\cite{FS96,R05}, worst-to-average case reductions~\cite{BLVW19,BCD23,DR25,MR04} and quantum reductions to the problem of finding short codewords and lattice points~\cite{DRT23a,R05}. 
\medskip

{\bf \noindent Alekhnovich and Regev' cryptosystem.} A public-key in Alekhnovich's encryption scheme is simply defined as an instance of the problem of decoding a random linear code. That is to say, a public random linear code $\C$, {\it i.e.,} a subspace of $\mathbb{F}_{2}^{n}$, and a noisy codeword $\vec{c} + \vec{e}$ where $\vec{c}\in \mathcal{C}$ and $\vec{e}$ being sparse, {\it i.e.,} with small Hamming weight. The associated secret-key is then defined as the solution of this decoding problem: the error $\vec{e}$. 
Then, to encrypt a single bit~$b \in \{0,1\}$, Alekhnovich proposed to proceed as follows: 
\medskip 
\begin{itemize}\setlength{\itemsep}{5pt}
	\item[$\bullet$] $\mathsf{Enc}(1)$ = $\vec{u}$ where $\vec{u}$ is a uniform vector;
	\item[$\bullet$] $\mathsf{Enc}(0)$ = $\vec{c}^{\perp}+\vec{e}'$ where $\vec{c}^{\perp}$ is a codeword of the {\em dual} (for the canonical inner product) of the code spanned by $\C$ and the noisy codeword $\vec{c} + \vec{e}$ while $\vec{e}'$ is as $\vec{e}$ a sparse vector.
\end{itemize}
\medskip 
To decrypt, an inner product is computed between the ciphertext and the secret-key. 
\medskip 
\begin{itemize}\setlength{\itemsep}{5pt}
	\item[$\bullet$] If $b=1$ has been encrypted, it is the inner product between a small weight vector and a uniform vector;
	\item[$\bullet$] If $b=0$ has been encrypted, it is the inner product between the secret-key and the vector~$\vec{e}'$ as $\vec{c}^{\perp}$ belongs to the dual of the code spanned by $\C$ and the noisy codeword $\vec{c} + \vec{e}$, therefore it belongs to the dual of the code spanned by~$\C$ and the secret-key~$\vec{e}$.
\end{itemize}
\medskip 
If $b=1$, the output value is a uniform bit. On the other hand, if $b=0$, the output is $0$ with high probability as it is the inner product of two sparse vectors. Then, repeating a small amount of times the previous process enables to recover the encrypted bit with overwhelming probability. 
Overall, two elements were critically used in Alekhnovich's encryption schemes: $(i)$ duality and~$(ii)$ the fact that the inner product of two sparse vectors is highly biased toward $0$.

While not efficient 
this scheme still marks a major breakthrough due to its security proof. The security only relies on distinguishing a uniform vector from a noisy codeword (in particular it does not require a ``structured'' decoding algorithm like in original McEliece's approach): problem that was shown to be equivalent to decoding a random code \cite{FS96}, 
the problem upon which the security of any code-based cryptographic scheme aims to be based on.

The same can be said about Regev's approach, the fundamental principle is that the inner product of two vectors of small coefficients modulo $q$ is kept ``small'' modulo $q$, and it still uses duality. It can thus be interpreted as being the same approach as that of Alekhnovich where the ``small Hamming weight'' is replaced by ``small coefficients modulo $q$''.

Of course, because Alekhnovich and Regev' schemes lacked efficiency, many improvements have been sought to make them practical. This started a very long line of work, introducing structures in the schemes to gain efficiency and performances at the price of security reductions, for instance \textsf{HQC} and \textsf{RQC}~\cite{DP12,DMN12,ABD+18,RQC,HQC} in the case of codes or \textsf{Frodo} and \textsf{Kyber} \cite{PVW08,KTX07,Kyber,Frodo} in the case of lattices.
Among these variants of Alekhnovich and Regev' schemes, a metric which is not Hamming or Euclidean has been considered with a certain success: the so-called {\em rank metric}.
\medskip

{\bf \noindent Matrix codes, rank metric and encryption schemes.} Introduced in 1978 by Delsarte~\cite{Delsarte}, matrix codes endowed with the rank metric are now fully part of the cryptographic landscape, either thanks to the now ubiquitous $\MinRank$ problem used for cryptanalysis~\cite{CG00,BFP11,FGP+15,B21,BTV22,SYY22,CMT23,GD24,SFI+25,STV25}, or thanks to the many cryptosystems using~$\Fqm$--linear codes~\cite{GPT,GHPT17,Loidreau,RQC}. Though known for more than $50$ years, building encryption schemes from the rank metric has often proven itself to be a very difficult task, either by following McEliece's framework, or by following Alekhnovich and Regev' approach.

A very promising line of work to design encryption schemes with matrix codes via McEliece's framework started with $\Fqm$--linear codes, codes which turn out to be structured matrix codes. The first family of $\mathbb{F}_{q^{m}}$--linear codes which has been proposed to design an encryption scheme {\em \`a la} McEliece were Gabidulin codes~\cite{G85}. This gave rise to the so-called \textsf{GPT} cryptosystem~\cite{GPT}. However, $\Fqm$--linearity and the particular structure of Gabidulin codes eased the cryptanalysis of this scheme and its many variations, broken by Overbeck in~\cite{O05,O08}. Many other schemes then followed this model, such as Loidreau's scheme~\cite{Loidreau} or \textsf{LowMS} \cite{ADG+}. A last scheme following McEliece's approach is \textsf{ROLLO} \cite{ROLLO}. It still uses $\mathbb{F}_{q^{m}}$--linear codes but not Gabidulin codes. It was instead proposed to use Low Rank Parity-Check  (LRPC) codes.

In 2016, Alekhnovich and Regev' framework was used in the rank metric, when Rank-Quasi-Cyclic (\textsf{RQC}) was introduced \cite{ABD+18}. The approach proved itself to be efficient, albeit at the cost of relying its security to the decoding problem of random $\mathbb{F}_{q^{m}}$--linear codes which are not generic matrix codes. 
A variation was also considered in~\cite{BBBG24}, which led to a very practical scheme. 
Over the years, \textsf{RQC} has been largely improved and has attracted the interest of the community \cite{BBBG24,ABD24,SCZ+}, making it fully part of the landscape of encryption schemes. In the meantime, in 2017, \textsf{RankPKE}, a scheme also relying on $\Fqm$--linearity via Alekhnovich's approach was introduced and used to build an identity-based encryption scheme \cite{GHPT17}. The security of this scheme relied on a variation of the problem to decode a random $\mathbb{F}_{q^{m}}$--linear code: 
the Rank Support Learning  (\textsf{RSL}) problem consisting of several decoding instances where the errors in the different noisy codewords are correlated. Overall, all these schemes attracted cryptanalysis, culminating in several attacks~\cite{AGHT18,BBB+20,BBC+20,BBB+23} that exploited the $\Fqm$--linearity, which has been used each time as part of the trapdoor in the aforementioned schemes.

A reader might now notice that among all these encryption schemes, not a single one used for its security the hardness of decoding a random matrix code, {\it i.e.,} $\MinRank$ problem, but only variations where the underlying codes to decode are particular matrix codes with an additional~$\mathbb{F}_{q^{m}}$--linear structure.
The lack of structure of the $\MinRank$ problem makes it a problem with a very strong security guarantee.
It is only very recently that the first encryption relying on $\MinRank$ was introduced. Following McEliece's framework, \cite{ACD24} managed to build a scheme starting from Gabidulin codes, and then removing their $\Fqm$--linear structure. However, this encryption scheme is very different from Alekhnovich and Regev' encryptions, as the security hypothesis are largely different. This leads us to the following question.
\begin{center}
	\textit{Can we build an encryption scheme following Alekhnovich and Regev' framework based solely \\ on $\minrank$ hardness, {\it i.e.,} the task of decoding a random~$\Fq$--linear matrix code?}
\end{center}
\medskip

{\bf \noindent Our contribution: an encryption scheme relying on stationary-$\MinRank$.} 
Our answer is mostly positive. We succeeded to design an encryption scheme, following Alekhnovich and Regev' framework, where we removed the $\mathbb{F}_{q^{m}}$--linear structure via the canonical  duality for matrix codes and the principle of ``small'' times ``small'' is ``small''. However, the security of our scheme is not directly based on $\minrank$ but on a slight variation:  stationary-$\minrank$ whose analogue in the Hamming metric case has been introduced in~\cite{KPR25}.

To explain how our scheme works and why we don't reduce its security directly to $\minrank$ hardness, one must first understand why natural adaptations of Alekhnovich and Regev' framework to the $\minrank$ setting is a priori doomed to failure. Roughly speaking, an adaptation of this approach fails when using $\MinRank$ as the inner product of two matrices of low rank (which is basically the trace of their product) has no reason to be biased toward $0$. The same goes if we take several $\MinRank$ instances. Taking $\ell_1$ instances in the public key, and $\ell_2$ in the ciphertext enables to build during decryption a $\ell_1 \times \ell_2$ matrix via different inner-products coming from small rank matrices. But then, there is no reason that the resulting matrix is a low rank matrix, as a priori the different inner products would not be related.

However, by taking slightly correlated instances, we actually obtain a working scheme thanks to the following (informal) theorem. 
\begin{theorem}[Informal]
	Let $\ell_1$ matrices $\vec{E}_{i}$'s such that their columns span the same space of small dimension $r$, and let $\ell_2$ matrices $\vec{F}_{j}$'s such that their rows span the same space of small dimension~$d$. Then, the matrix composed of all the~$\ell_1 \times \ell_2$ inner products $\langle \vec{E}_{i},\vec{F}_{j} \rangle$  is of dimension~$\le rd$.
\end{theorem}

We use this theorem as the core result to design our encryption scheme. The public-key consists now in taking $\ell_1$ instances $\vec{C}_{i}+\vec{E}_i$'s of $\minrank$ where matrices~$\vec{E}_{i}$'s (which are the secret-key) are such that their columns {\em span the same space}. Notice that recovering the secret-key from the public-key does not amount to solve $\minrank$ where independent instances are given, but a variation where different errors are correlated, it is the stationary-$\minrank$ problem. It roughly explains why the security of our scheme does not reduce to $\minrank$. Now to encrypt a bit we proceed as in Alekhnovich and Regev' framework, instead that to encrypt~$b=0$ we produce $\ell_{2}$ instances $\vec{D}_{j}+\vec{F}_{j}$'s of $\minrank$  where the different errors $\vec{F}_{j}$'s are such that their rows {\em span the same space}. By using once again the canonical duality approach, during decryption we compute a list of~$\ell_1 \times \ell_{2}$ inner-products~$\langle \vec{E}_{i},\vec{F}_{j} \rangle$ which in that case gives a small rank matrix! On the other hand, when~$b=1$ has been encrypted, our inner-products are derived from uniform matrices which typically gives a matrix with large rank.

We prove (following Alekhnovich and Regev' proof) that our scheme relies on the decisional version of stationary-$\MinRank$, and give a search-to-decision reduction. As a result, the security of our scheme relies solely on the hardness of stationary-$\MinRank$ problem in its search version, avoiding any additional security assumptions.

Finally, we carefully analyze attacks on our scheme.
We find that these attacks all come down to solving a $\MinRank$ instance, which strongly reinforces the view that the security of our scheme is close to decoding a random matrix code. This enabled us to propose parameter sets for our scheme. We obtain sizes of public-keys and ciphertexts of around~$14$kB each, for a total of~$28$kB. While larger than schemes such as \textsf{HQC}, \textsf{RQC}, or \textsf{Kyber} (now \textsf{ML-KEM}), it is comparable to \textsf{FrodoKEM}, with bigger but reasonable parameters. Finally, our work also opens the door to many other questions, in particular concerning improving the efficiency of the scheme, in the same vein as \textsf{HQC}, \textsf{RQC} and \textsf{Kyber} can be seen as improvements of Alekhnovich' and Regev's encryption schemes.
\medskip

{\bf \noindent Organization of the paper.} We begin by describing some notation and background on matrix codes in Section \ref{sec:preliminaires}. Then, we describe our construction to encrypt one bit in Section \ref{sec:Encstructure}, followed by several bits in Section \ref{sec:MRPKE}. Finally, we detail our considered attacks against our scheme in Section~\ref{sec:security}, which allows us to offer concrete parameters for $\lambda$ bits of security. Parameters, performances, and comparison to other schemes are provided in Section~\ref{sec:parameters}. \section{Notation and Matrix Codes Background} \label{sec:preliminaires}

{\bf \noindent Basic notation.} The notation $x \eqdef y$ means that
$x$ is being defined as equal to~$y$. Let $a < b$ be integers, we let
$[a,b]$ to denote the set of integers $\{a,a+1,\dots,b\}$.

Vectors are in row notation and they will be written with bold letters
such as $\vec{a}$. Uppercase bold letters such as $\vec{A}$ are used
to denote matrices. Let $q$ be a power of a prime number. We let
$\mathbb{F}_{q}$ to denote the finite field of cardinality $q$. Given
integers $m,n$, we let $\mathbb{F}_{q}^{m \times n}$ denote the set of
matrices with $m$ rows and $n$ columns whose entries belong to
$\mathbb{F}_{q}$. Given $\vec{M} \in \mathbb{F}_{q}^{m \times n}$, we let~$\vec{M}(i,j)$ denote its coefficient in position $(i,j)$ while $\mathbf{Col}(\vec{M},j) \in \mathbb{F}_{q}^{m}$ denotes column $j\in [1,n]$ of~$\vec{M}$. Furthermore, we let $\Sp{\vec{M}}$ denote the subspace of $\mathbb{F}_{q}^{m}$ spanned by {\it columns} of $\vec{M}$, {\it i.e.,}
$$
\Sp{\vec{M}} \eqdef \left\{ \sum_{j=1}^{n} \lambda_{j} \mathbf{Col}(\vec{M},j): \; \lambda_{j} \in \mathbb{F}_{q} \right\} . 
$$
We call $\Sp{\vec{M}}$ the {\em column support} of $\vec{M}$. The {\em row support} of $\vec{M}$ is then defined as~$\Sp{\vec{M}^{\top}}$.
Given vectors $\vv_1, \dots, \vv_s$ in a given space, we let
$\mathbf{Span}\left( \vv_1, \dots, \vv_s\right)$ to denote the subspace they span. In particular, given $\vec{B}_1,\dots,\vec{B}_{k} \in \mathbb{F}_{q}^{m \times n}$, notice that $\mathbf{Span}\left(  \vec{B}_{1},\dots,\vec{B}_{k}\right)$ is a $\mathbb{F}_{q}$-subspace of~$\mathbb{F}_{q}^{m \times n}$.

In what follows,
$\mathcal{B}_{t}^{m,n,q}$, 
will denote the ball 
of radius~$t$ around $\vec{0}_{m\times n}$
in~$\mathbb{F}_{q}^{m\times n}$ for the rank metric $|\cdot |$ which
is defined as 
$$
\vec{A} \in \mathbb{F}_{q}^{m \times n}, \; |\vec{A}| \eqdef
\rk(\vec{A}) \ .
$$

We will consider the following canonical inner product over $\mathbb{F}_{q}^{m \times n}$,  
$$
\forall \vec{A},\vec{B}\in \mathbb{F}_{q}^{m \times n}, \; \langle \vec{A}, \vec{B} \rangle \eqdef \tr\left( \vec{A}\vec{B}^{\top} \right).
$$

	{\bf \noindent Matrix codes.} A \emph{matrix code} $\mathcal{C}$ over
	$\mathbb{F}_{q}$ with length $m \times n$ and dimension~$k$ is a
	subspace of dimension $k$ of the vector space
	$\mathbb{F}_{q}^{m \times n}$. We say that it is an~$[m \times n,k]_{q}$-code.  
	Given an $[m \times n,k]_{q}$-code $\mathcal{C}$, its {\em dual} (relatively to the above inner product) is defined as 
	$$
	\mathcal{C}^{\perp}\eqdef \left\{ \vec{B}^{\perp}\in \mathbb{F}_{q}^{m
		\times n}\!\!: \; \forall \vec{C} \in \mathcal{C}, \; \langle 
	\vec{B}^{\perp}, \vec{C} \rangle = 0 \right\}.
	$$
	It defines an $[m \times n, mn-k]_{q}$-code. Furthermore, if $\vec{B}^{\perp}_{1},\dots,\vec{B}^{\perp}_{mn-k}$ denotes a basis of $\mathcal{C}^{\perp}$, then 
	$$
	\mathcal{C} = \left\{ \vec{C} \in \mathbb{F}_{q}^{m\times n}\!\!: \; \forall i \in [1,mn-k], \; \langle  \vec{B}_{i}^{\perp},\vec{C} \rangle= 0 \right\}.
	$$

	{\bf \noindent Probabilistic notation.} For a finite set $\mathcal{E}$, we write $X \Unif \mathcal{E}$ when $X$ is an element of $\mathcal{E}$ drawn uniformly at random. 
	
 \section{$\minrank$-based Encryption \`a la Alekhnovich-Regev} \label{sec:Encstructure}

Our focus in this paper is to design a rank-based encryption scheme following Alekhnovich~\cite{A03} and Regev'~\cite{R05} approach. Our aim is therefore to design an encryption scheme whose security is {\em only} based on the average hardness of the $\minrank$ problem, which is stated in its {\it primal} form as follows. It consists in decoding a random matrix-code, {\em i.e.,} a matrix-code sampled
uniformly at random.

\begin{definition}[$\minrank$, {\it primal representation}]\label{def:minRankPF} Let $m,n,k,t,q$ be integers that are functions of some security
	parameter $\lambda$ and such that $mn \geq k$. Let $\vec{E} \in \mathcal{B}_{t}^{m,n,q}$, $\vec{B}_1, \dots, \vec{B}_{k} \in \mathbb{F}_{q}^{m \times
		n}$ and~$\lambda_1,\dots,\lambda_{k} \in \mathbb{F}_{q}$ be sampled uniformly at
	random. Let,
	$$
	\vec{Y} \eqdef \sum_{\ell=1}^{k} \lambda_{\ell} \vec{B}_{\ell} + \vec{E} \ .
	$$
	The $\minrank(m,n,k,t,q)$
	problem consists, given $\left( \vec{B}_{1},\dots,\vec{B}_{k},\vec{Y}\right)$, in finding~$\vec{E}$.
\end{definition}

\begin{remark}Notice that given $\sum_{\ell} \lambda_{\ell} \vec{B}_{\ell} + \vec{E}$, we ask to recover $\vec{E}$ which has rank $\leq t$. Therefore, in order for this problem to be meaningful, parameter $t$ has to be small enough to ensure  with overwhelming probability the unicity of $\vec{E}$. It is sufficient and necessary to choose $t$
	below the so-called {\em Gilbert-Varshamov} radius which is for instance when $m=n$ asymptotically equivalent~\cite{courtois,L06} to~$
m \left( 1 - \sqrt{\frac{k}{m^{2}}} \right) \ .
	$
\end{remark}

It turns out that $\minrank$ also admits an equivalent {\em dual} form. By equivalent, we mean that from any solver (on average) of the primal representation of $\minrank$, we can deduce a solver of its dual form with the same probability of success (up to an exponentially small factor) and working with the same amount of time (up to polynomial factors) and reciprocally. The key fact to prove this result is that given a random matrix code\footnote{Random matrix codes are defined as matrix codes whose basis has been sampled uniformly at random.}, then its dual is still a random matrix code.  The dual version of $\minrank$ is stated as follows. We have chosen to state it as we will be switching between both forms throughout the paper.

\begin{definition}[$\minrank$, {\it dual representation}]\label{def:minRankDF} Let $m,n,k,t,q$ be integers that are functions of some security
	parameter $\lambda$ and such that $mn \geq k$. Let $\vec{E} \in \mathcal{B}_{t}^{m,n,q}$  and
	$\vec{B}^{\perp}_1, \dots, \vec{B}^{\perp}_{mn - k} \in \mathbb{F}_{q}^{m \times
		n}$ be sampled uniformly at
	random and
	$$
	\vec{s} \eqdef \left( \langle \vec{B}_{\ell}^{\perp}, \vec{E} \rangle
	\right)_{\ell= 1}^{ mn-k} .
	$$
	The $\minrank(m,n,k,t,q)$
	problem consists, given
	$(\vec{B}^{\perp}_{1},\dots,\vec{B}^{\perp}_{mn-k},\vec{s})$, in finding~$\vec{E}$.
\end{definition}
\medskip 

{\bf \noindent Encrypting one bit via $\minrank$ hardness.}
Alekhnovich~\cite{A03} and Regev'~\cite{R05} encryptions are both based on the fact that public keys are defined as an instance of decoding a random linear code (for codes endowed with the Hamming metric) and $\mathsf{LWE}$ problems while secret-keys are their associated solution. It is therefore tempting to define a public-key of our scheme as~$(\vec{B}_{\ell})_{\ell}$,~$\vec{Y} \eqdef \sum_{\ell} \lambda_{\ell} \vec{B}_{\ell} + \vec{E}$ and its associated secret-key as~$\vec{E}$. By doing so, if one wishes to encrypt one bit following \cite{A03,R05}, one proceeds as follows.
\medskip 
\begin{itemize}\setlength{\itemsep}{5pt}
	\item To encrypt $b = 1$: output $\vec{U}$ a uniform matrix with the same size than $\vec{Y}$,

	\item To encrypt $b = 0$: output $\vec{C}^{\perp}+\vec{F}$ where $\vec{C}^{\perp}$ has been sampled uniformly at random in the {\em dual} of the code spanned by the public key, {\it i.e.,} $\left( \vec{B}_{\ell} \right)_{\ell}$ and~$\vec{Y}$.
\end{itemize}
\medskip  
Given a cipher $\vec{Z}$, to decrypt we compute the following inner-product  
$$
\langle \vec{Z},\vec{E} \rangle = \left\{ \begin{array}{cl} 
	\langle \vec{U},\vec{E} \rangle & \mbox{ if $b=1$} \\
	\langle \vec{F}, \vec{E} \rangle & \mbox{ if $b=0$} 
\end{array}\right.  
$$ 
where in the second equality we used that $\vec{C}^{\perp}$ belong to the {\em dual} of the matrix code spanned by~$\left( \vec{B}_{\ell} \right)_{\ell}$ and $\vec{Y} = \sum_{\ell} \lambda_{\ell} \vec{B}_{\ell}+\vec{E}$, therefore it belongs to the {\em dual} of the matrix code spanned by~$\left( \vec{B}_{\ell} \right)_{\ell}$ and~$\vec{E}$. Alekhnovich and Regev' idea is that this inner-product should have a distribution strongly correlated to  the encrypted bit. In particular, it should be uniform when $b=1$ has been encrypted and ``small'' otherwise. By small we mean that repeating the operation a certain number of times produces a vector with small norm. However, in our case ``small'' should mean a matrix with small rank.

This discussion motivated us to introduce the variant of Alekhnovich and Regev' encryption where a list of $\ell_1\geq 1$ instances of $\minrank$ are given as public-key 
$$
\left( \left( \vec{B}_{\ell}^{(j)} \right)_{\ell}, \vec{Y}^{(j)} \eqdef \sum_{\ell} \lambda_{\ell} \vec{B}_{\ell}^{(j)} + \vec{E}^{(j)}  \right)_{j \in [1,\ell_1]} \ .
$$ 
The secret-key is then the collection of the $\vec{E}^{(j)}$'s. Now to encrypt a bit we simply proceed as above, but this times we repeat the process $\ell_2\geq 1$ times according to the bit we wish to encrypt. For instance, to encrypt $b = 0$, we output
$$
\left( \vec{Z}^{(i)} \eqdef \vec{C}_{i}^{\perp} + \vec{F}^{(i)} \right)_{i \in [1,\ell_2]}
$$
where the $\vec{C}_{i}^{\perp}$'s are sampled uniformly at random in the dual of a matrix-code (exhibited below) which is obtained from the public-key. Now to decrypt we compute the list of inner-products~$\langle \vec{Z}^{(i)},\vec{E}^{(j)} \rangle$'s for~$j \in [1,\ell_1]$ and~$i \in [1,\ell_2]$ to form {\em a matrix} of size $\ell_2 \times \ell_1$. They are for instance given by 
$$
\left( \langle \vec{C}_{i}^{\perp},\vec{E}^{(j)} \rangle + \langle \vec{F}^{(i)},\vec{E}^{(j)} \rangle \right)_{i \in [1,\ell_2], j \in [1,\ell_1]}
$$
in the case where $b = 0$ has been encrypted. However, notice that the $\langle \vec{C}_{i}^{\perp},\vec{E}^{(j)} \rangle$'s have no reason to vanish like in Alekhnovich and Regev' encryption. It is why if one wishes to encrypt $b = 0$, one has to draw uniformly at random $\vec{C}_{i}^{\perp}$ in the {\em dual} of the code spanned by 
$$
\left( \left( \vec{B}_{\ell}^{(j)} \right)_{\ell}, \vec{Y}^{(j)} \right)_{j \in [1,\ell_1]} \ .
$$ 
In particular, $\vec{C}_{i}^{\perp}$ belongs to the dual of the sum over $j$ of the codes spanned by $\left( \vec{B}_{\ell}^{(j)}\right)_{\ell}$ and the~$\vec{Y}^{(j)}$'s which are given by the public-key. Notice that this (sum) code contains all the $\vec{E}^{(j)}$'s, {\it i.e.,} the secret-key. Therefore the~$\langle \vec{C}_{i}^{\perp},\vec{E}^{(j)} \rangle$'s are all equal to $0$ for $j \in [1,\ell_1]$. In other words, during decryption, where we compute the list of inner-products~$\left( \langle \vec{Z}^{(i)}, \vec{E}^{(j)} \rangle\right)$'s, we obtain the following matrix
$$
\left\{ \begin{array}{cl} 
	\left( \langle \vec{U}^{(i)},\vec{E}^{(j)} \rangle \right)_{i \in [1,\ell_2], j \in [1,\ell_1]} & \mbox{ if $b=1$} \\
	\left( \langle \vec{F}^{(i)}, \vec{E}^{(j)} \rangle\right)_{i \in [1,\ell_2], j \in [1,\ell_1]} & \mbox{ if $b=0$} 
\end{array}\right.  
$$  
where the $\vec{U}^{(i)}$'s are uniform matrices. In particular, when $ b = 1$ has been encrypted, we typically obtain a full-rank matrix. On the other hand, following Alekhnovich and Regev' idea, we should obtain a matrix with small rank when~$b = 0$ has been encrypted. But there are no reason to achieve this. Following once again Alekhnovich and Regev' approach leads to choose the $\vec{E}^{(j)}$'s and~$\vec{F}^{(i)}$'s with {\em small rank}. But this does {\em not} imply that $\left( \langle \vec{F}^{(i)}, \vec{E}^{(j)} \rangle \right)_{i, j}$ is typically a small rank matrix. 

Our discussion seems to suggest that following  Alekhnovich and Regev' approach in the $\minrank$ case is doomed to failure. But, surprisingly, a small variation enables to ensure a small rank matrix when $b = 0$ has been encrypted. As shown in Theorem~\ref{theo:fund}, if the $\vec{E}^{(j)}$'s have small rank and {\em the same column support} and in addition, the $\vec{F}^{(i)}$'s also have a small rank and {\em the same row support}, then $\left( \langle \vec{F}^{(i)},\vec{E}^{(j)} \rangle \right)_{i, j}$ {\em is} a small rank matrix!

Therefore, we just need to impose a constraint on column and row supports of errors during key generation and encryption to achieve Alekhnovich and Regev' approach with matrix codes.   However, doing so means that security no longer relies on $\minrank$ average hardness. But as discussed later (see Theorem~\ref{theo:secu}), the security still reduces to one search problem which turns out to be a slight variation of $\minrank$: the so-called {\it stationary}-$\minrank$ problem given in Definition~\ref{def:stMinrank}. It roughly consists in $\minrank$ where we are given multiple instances with independent random matrix codes but with errors sharing the same unknown column support.

\begin{theorem}\label{theo:fund}
	Let $q,m,n,r,d,\ell_1,\ell_2$ be integers such that $m\geq n > r \geq d$. Let~$A\subseteq \mathbb{F}_{q}^{m}$ and~$B \subseteq \mathbb{F}_{q}^{n}$ be two subspaces with dimensions $r$ and $d$ respectively. Let~$\vec{A}_1,\dots, \vec{A}_{\ell_1} \in \mathbb{F}_{q}^{m\times n}$ and~$\vec{B}_{1},\dots, \vec{B}_{\ell_2} \in \mathbb{F}_{q}^{m \times n}$ such that 
	$$
	\forall j \in [1,\ell_1], \;\; \Sp{\vec{A}_{j}} = A \quad \mbox{and} \quad \forall i \in [1,\ell_2], \;\;  \Sp{\vec{B}_{i}^{\top}} = B \ .
	$$ 
	Let, $\vec{D} \in \mathbb{F}_{q}^{\ell_2 \times \ell_1}$ such that $\vec{D}(i, j) \eqdef \langle \vec{A}_{j},\vec{B}_{i}\rangle$. Then,
	$$
	| \vec{D}| \leq \min\left(rd, \ell_1,\ell_2 \right) \ .
	$$
\end{theorem}
\begin{proof} 
	By definition it exists $\vec{A} \in \mathbb{F}_{q}^{m\times r}$ and $\vec{B} \in \mathbb{F}_{q}^{n \times d}$ such that 
$$
\forall j \in [1,\ell_1], \; \exists \vec{P}_{j} \in \mathbb{F}_{q}^{r \times n}:\; \vec{A}_{j} = \vec{A}\vec{P}_{j}, \quad  
\forall i \in [1,\ell_2], \; \exists \vec{Q}_{i} \in \mathbb{F}_{q}^{d \times m}: \; \vec{B}_{i}^{\top} = \vec{B}\vec{Q}_{i}
$$
We deduce that,
	$$
	\vec{D}(i,j)=\tr\left( \vec{A}_{j}\vec{B}_{i}^{\top} \right) = \tr\left( \vec{A}\vec{P}_{j}\vec{B}\vec{Q}_{i} \right) = \tr\left(\vec{P}_{j}\vec{B}\vec{Q}_{i}\vec{A} \right) 
	$$
	where in the last equality we used the cyclicity of the trace. 
	Notice now that $\vec{P}_{j}\vec{B} \in \mathbb{F}_{q}^{r \times d}$. Therefore, it exists at most $rd$ matrices $\vec{P}_{j}$'s such that the $\vec{P}_{j}\vec{B}$'s are linearly independent. Suppose now that $|\vec{D}| > rd$ and~$\ell_2 > rd$. Then it exists $rd+1$ columns which are linearly independent:
	$$
	\mathbf{Col}(\vec{D},j_1) = \tr\left(\vec{P}_{j_1}\vec{B}\vec{Q}_{i}\vec{A} \right)_{i} , \dots, \mathbf{Col}(\vec{D}, j_{rd+1}) = \tr\left(\vec{P}_{j_{rd+1}}\vec{B}\vec{Q}_{i}\vec{A} \right)_{i} \ .
	$$
	But we know that it exists a non-zero $(\lambda_{j_1},\dots,\lambda_{j_{rd+1}}) \in \mathbb{F}_{q}^{rd+1}$ such that $\sum_{a = 1}^{rd+1} \lambda_{j_1}\vec{P}_{j_a}\vec{B} = \vec{0}$ which contradicts the linear independence of the $\mathbf{Col}(\vec{D}, j_a)$'s.  
The same reasoning holds for rows of $\vec{D}$ which concludes the proof.\iftoggle{llncs}{\qed}{}
\end{proof}

Our public-key encryption scheme is described in Figure~\ref{algo:scheme}. In Proposition~\ref{propo:correctDec} we show that decryption is successful with overwhelming probability if $rd$ is sufficiently small compared to~$\ell_1$ and $\ell_2$. 

\begin{figure}[h!]
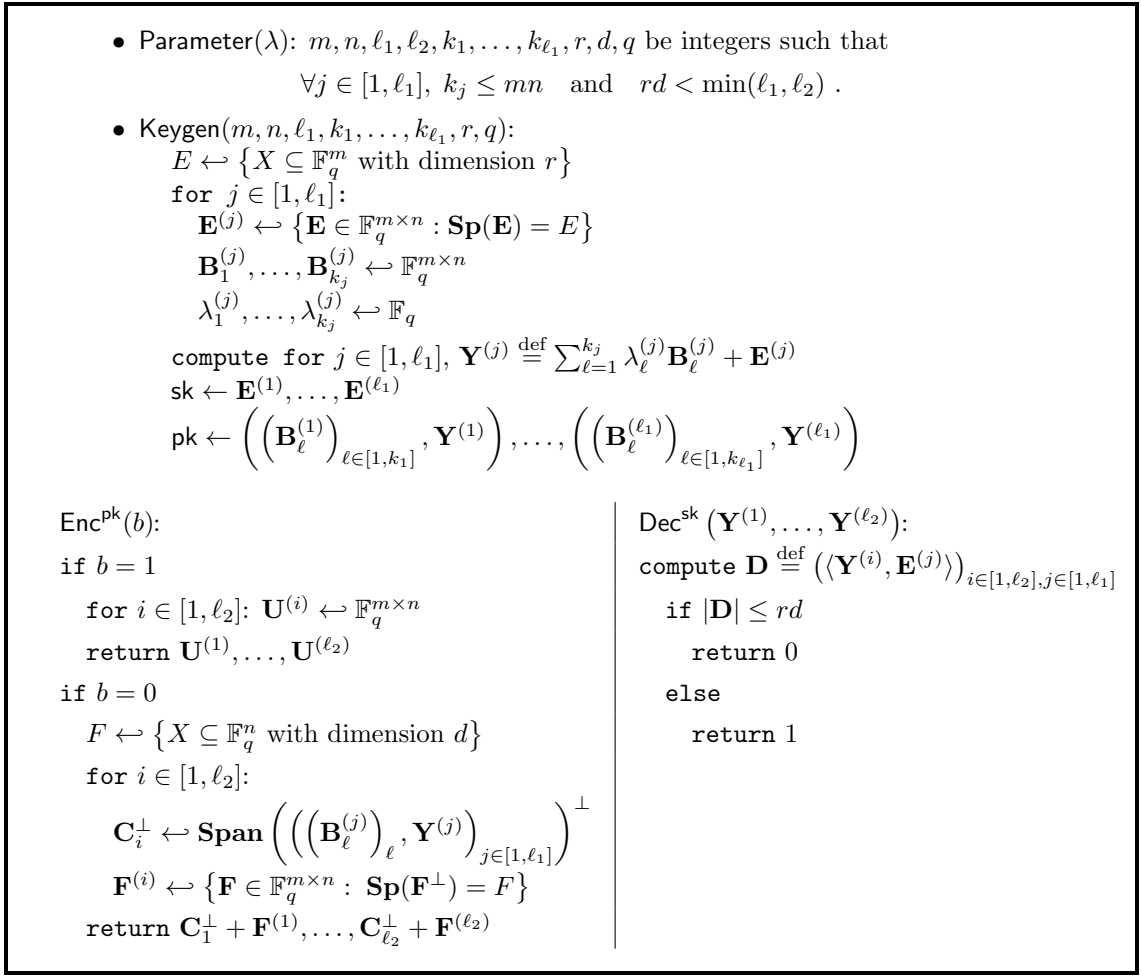

	
	\begin{tcolorbox}[colframe=black, colback=white, boxrule=0.5mm, sharp corners]
		\begin{itemize}\setlength{\itemsep}{5pt}
			\item $\mathsf{Parameter}(\lambda)$: $m,n,\ell_1,\ell_2, k_1,\dots,k_{\ell_1},r,d,q$ be integers such that 
			$$
			\forall j \in [1,\ell_1], \; k_{j} \leq mn \quad \mbox{and} \quad rd < \min(\ell_1,\ell_2) \ .
			$$

			\item$\mathsf{Keygen}(m,n,\ell_1, k_1,\dots,k_{\ell_1}, r, q)$:

			$E \Unif \left\{ X \subseteq \mathbb{F}_{q}^{m} \mbox{ with dimension $r$} \right\}$
			
			\texttt{for $j \in [1,\ell_1]$:}

			\quad $\vec{E}^{(j)} \Unif \left\{ \vec{E} \in \mathbb{F}_{q}^{m\times n} : \Sp{\vec{E}} = E \right\}$

			\quad $\vec{B}_1^{(j)},\dots, \vec{B}_{k_{j}}^{(j)} \Unif \mathbb{F}_{q}^{m \times n}$

			\quad $\lambda_{1}^{(j)},\dots,\lambda_{k_{j}}^{(j)} \Unif \mathbb{F}_{q}$

			\texttt{compute for} $j \in [1,\ell_1]$, $\vec{Y}^{(j)} \eqdef  \sum_{\ell=1}^{k_{j}} \lambda_{\ell}^{(j)}\vec{B}_{\ell}^{(j)} + \vec{E}^{(j)}$
			
			$\mathsf{sk} \leftarrow \vec{E}^{(1)},\dots,\vec{E}^{(\ell_1)}$

			$\mathsf{pk} \leftarrow \left( \left( \vec{B}_{\ell}^{(1)}\right)_{\ell \in [1,k_1]}, \vec{Y}^{(1)} \right), \dots,  \left( \left( \vec{B}_{\ell}^{(\ell_1)}\right)_{\ell \in [1,k_{\ell_1}]}, \vec{Y}^{(\ell_1)} \right)$
		\end{itemize}
		\medskip

		\begin{center}
			\renewcommand{\arraystretch}{1.3}
			\begin{tabular}{l@{\hspace{3mm}}|@{\hspace{3mm}}l}
				$\mathsf{Enc}^{\mathsf{pk}}(b)\!\!: \qquad \qquad \qquad$ & $\mathsf{Dec}^{\mathsf{sk}}\left(\vec{Y}^{(1)},\dots,\vec{Y}^{(\ell_2)}\right)\!\!:$ \\
				\texttt{if} $b = 1$ & \texttt{compute} $\vec{D} \eqdef \left( \langle \vec{Y}^{(i)},\vec{E}^{(j)} \rangle\right)_{i\in [1,\ell_{2}],j\in [1,\ell_1]}$ \\
\quad \texttt{for} $i\in [1,\ell_2]$: $\vec{U}^{(i)}\Unif \mathbb{F}_{q}^{m \times n}$  & \quad \texttt{if} $|\vec{D}| \leq rd$\\ 
				\quad \texttt{return} $\vec{U}^{(1)},\dots,\vec{U}^{(\ell_2)}$ &  \qquad \texttt{return} 0 \\
				\texttt{if} $b=0$ & \quad \texttt{else}		\\	
			\quad $F \Unif \left\{ X \subseteq \mathbb{F}_{q}^{n} \mbox{ with dimension $d$} \right\}$ & \qquad \texttt{return} 1 \\ 
				\quad \texttt{for} $i \in [1,\ell_2]$:&   \\
				\qquad $\vec{C}_i^{\perp} \Unif \mathbf{Span} \left( \left( \left( \vec{B}_{\ell}^{(j)}\right)_{\ell}, \vec{Y}^{(j)} \right)_{j \in [1,\ell_1]}\right)^{\perp}$ & \\

				\qquad $\vec{F}^{(i)} \Unif \left\{ \vec{F} \in \mathbb{F}_{q}^{m \times n}: \; \Sp{\vec{F}^{\perp}} =F \right\}$ & \\
				\quad\texttt{return} $\vec{C}_1^{\perp} + \vec{F}^{(1)}, \dots, \vec{C}_{\ell_2}^{\perp} + \vec{F}^{(\ell_2)}$ & 
			\end{tabular} 
\end{center} 
	\end{tcolorbox}
	\caption{Our first $\minrank$-based public-key encryption scheme.}\label{algo:scheme}
\end{figure} 

\begin{proposition}[Correctness of decryption]\label{propo:correctDec} 
	Decryption in 
Figure \ref{algo:scheme} fails with probability 
	$$
	O\left( q^{\min\left( rd+1-\ell_2, r(\ell_1-n)\right)} \right)
	$$
\end{proposition}

\begin{proof}
	Let $\vec{D} \in \mathbb{F}_{q}^{\ell_1 \times \ell_2}$ such that 
	$
	\forall i \in [1,\ell_2], \; \forall j \in [1,\ell_1], \; \vec{D}(i,j) \eqdef \langle \vec{U}^{(i)}, \vec{E}^{(j)} \rangle 
	$
	where the $\vec{U}^{(i)}\in \mathbb{F}_{q}^{m \times n}$ are independent and uniformly distributed and the $\vec{E}^{(j)}$s are fixed matrices computed during keys generation. Notice that if decryption fails, then it means that $\vec{D}$ has rank~$\leq rd < \min(\ell_1,\ell_{2})$ (if $b = 0$ has been encrypted, decryption is always successful as ensured by Theorem~\ref{theo:fund}). Therefore, decryption fails if it exists non-zero $(\lambda_1, \dots, \lambda_{rd+1}) \in \mathbb{F}_{q}^{rd+1}$ such that 
	$$
	\sum_{j=1}^{rd+1} \lambda_j\mathbf{Col}(\vec{D},j) = \sum_{j=1}^{rd+1} \lambda_{j} \left( \langle \vec{U}^{(i)}, \vec{E}^{(j)} \rangle \right)_{i \in [1,\ell_2]} =  \vec{0} 
	$$
	which implies by bilinearity of the inner product,
	$$
	\forall i \in [1,\ell_2], \; \langle \vec{U}^{(i)}, \sum_{j=1}^{rd+1}\lambda_{j}\vec{E}^{(j)} \rangle = 0 \ . 
	$$
	Let us suppose that the $\vec{E}^{(j)}$'s are linearly independent. The above equations for $i \in [1,\ell_2]$ and any non-zero $(\lambda_1, \dots, \lambda_{rd+1}) \in \mathbb{F}_{q}^{rd+1}$ are verified with probability~$
	\frac{1}{q^{\ell_2}}
	$ 
	as the $\vec{U}_{i}$'s are independent and uniform. By union bound we deduce that decryption fails with probability $\leq q^{rd+1-\ell_2}$ if the $\vec{E}_{j}$'s are linearly independent. Notice now that it happens with probability~$1-O\left( q^{r(\ell_1-n)} \right)$ during key generation. It concludes the proof.\iftoggle{llncs}{\qed}{} 
\end{proof}

{\bf \noindent About the security.} One may naturally argue that the security of our scheme in Figure~\ref{algo:scheme} does not a priori reduce to $\minrank$ as shown by the simple fact that computing a secret-key from a public-key does not amount to solve $\minrank$. Indeed, such computation reduces to solve the following {\it stationary}-$\minrank$ problem which has been introduced in Hamming metric in~\cite{KPR25}.

\begin{definition}[stationary-$\minrank$]\label{def:stMinrank} Let $m,n, N, k_1,\dots,k_N,t,q$ be integers which are functions of some security
	parameter $\lambda$ and such that $mn \geq k_j$ for all~$j \in [1,N]$.  Let
	$\vec{B}^{(j)}_1, \dots, \vec{B}^{(j)}_{k_j} \in \mathbb{F}_{q}^{m \times
		n}$ and~$\lambda^{(j)}_1,\dots,\lambda^{(j)}_{k_j} \in \mathbb{F}_{q}$ for $j \in [1,N]$ be sampled uniformly at
	random. Let $E \subseteq \mathbb{F}_{q}^{m}$ be a random subspace with dimension~$t$ and $\vec{E}^{(j)} \in \left\{ \vec{X} \in \mathbb{F}_{q}^{m \times n}: \; \Sp{\vec{X}}= E \right\}$ be uniformly distributed  for~$j \in [1,N]$. Let,
	$$
	\forall j \in [1,N], \; \vec{Y}^{(j)} \eqdef \sum_{i=1}^{k_j} \lambda^{(j)}_{i} \vec{B}^{(j)}_{i} + \vec{E}^{(j)} .
	$$
	The stationary-$\minrank(m,n,N,(k_i)_{i \in [1,N]},t,q)$
	problem consists, given $\left( \left( \vec{B}^{(j)}_{\ell}\right)_{\ell \in [1,k_j]},\vec{Y}^{(j)}\right)_{j}$, in finding $E$.
\end{definition}

\begin{remark}
	A {\it stationary}-$\minrank$ instance consists in multiple $\minrank$ instances where we are given noisy codewords from independent random matrix codes but with correlated errors.  
\end{remark}

Though breaking our encryption scheme by computing the secret-key from the knowledge of the public-key amounts to solve the above (slight) variation of $\minrank$, it does not show that
our scheme security reduces to it. In particular, an attacker could basically seek to distinguish uniform matrices from multiple $\vec{C}_{i}^{\perp} + \vec{F}_{i}$'s where the $\vec{F}_{i}$'s have the same row support and the~$\vec{C}_{i}^{\perp}$'s are codewords from the {\em same} code. Fortunately, the security of our scheme reduces to stationary-$\minrank$ as shown by the following theorem. 
\begin{theorem}\label{theo:secu}
	Consider an attacker against the public-key encryption scheme described in Figure~\ref{algo:scheme}. Suppose that this attacker extracts an encrypted bit in time $T$ with probability $1/2+\varepsilon$ by using knowledge of the public-key.

	Then, there exists an algorithm which solves stationary-$\minrank$ for parameters (with $\sum_{j=1}^{\ell_1} k_j+\ell_1= \sum_{i=1}^{\ell_2} k'_i$)
	$$
	(m,n,\ell_1,(k_i)_{i \in [1,\ell_1]},r,2)\quad \mbox{ or } \quad (m,n,\ell_{2},(mn/\ell_2-k'_i)_{i \in [1,N]},d,2) 
	$$ 
	working in time $O\left( \ell(mn)^{2} \log_2^{3}(\ell/\varepsilon) T\right)$ with probability $\Omega\left(\frac{\varepsilon^{2}}{\ell^{2}}\right)$ where $\ell \eqdef \max(\ell_1,\ell_2)$. 
\end{theorem}

\begin{remark}
	Notice that we have fixed $q=2$. In our instantiations we will each time choose $q$ as being equal to $2$.  
\end{remark}

To prove this theorem we proceed in basically two steps. 
\medskip 
\begin{itemize}\setlength{\itemsep}{5pt}
	\item[{\bf Step 1.}] First, we introduce the {\it decisional} version of stationary-$\minrank$ which consists in the problem of distinguishing between a true stationary-$\minrank$ instance and random matrices with the same sizes. We then show that any attacker against our scheme can be turned into an algorithm solving the decisional stationary-$\minrank$ problem.

	\item[{\bf Step 2.}]We end the proof by proving that decisional stationary-$\minrank$ is harder that it search counterpart via a {\em search-to-decision} reduction using Goldreich-Levin theorem as subroutine.
\end{itemize}
\medskip 

{\bf \noindent Proof of Theorem~\ref{theo:secu}.} Our security reduction is nothing new compared to the security reduction of  Alekhnovich and Regev' encryption schemes. We follow exactly the same strategy. Our contribution is mainly to provide a search-to-decision reduction for the stationary-$\minrank$ problem. Our reduction basically follows the standard search-to-decision reduction (with an additional hybridization trick) for the decoding problem of random codes endowed with the Hamming metric from \cite{FS96}. The full proof can be found in Appendix~\ref{app:reduction}.

 \section{MinRankPKE} \label{sec:MRPKE}

The scheme presented in the previous section is not efficient. Taking secure parameters would roughly give us a public-key size of more than $6$MB and a ciphertext size of more than $4$MB. The main reason we are getting such poor parameters is that the encryption rate is particularly low, {\it i.e.,} $1/(\ell_{2} mn)$.

The aim of this section is to make our scheme efficient without affecting its security reduction. We do so by describing the encryption when we directly encrypt many bits at once instead of only one bit. Our proposed construction is not new, it consists in following exactly the same approach as with Alekhnovich and Regev' schemes to encrypt several bits all at once (see for instance \cite[\S4.4, Cryptosystem $2$]{A03}). In particular, we will make use of an efficient decodable code. Let us stress that it is known to {\em not} affect the security of the approach.  
We also describe the scheme instantiated by taking only one matrix code as a public-key instead of $\ell_1$ different small codes. Notice that doing so does not affect the security reduction. If an attacker can break the scheme with one code as a public-key, then it can break the case where $\ell_1$ codes are given. It is enough when trying to break the scheme with $\ell_1$ matrix codes to consider as public-key the sum of these codes and to feed this to the attacker. 
\medskip

{\bf \noindent Encrypting more than one bit with Gabidulin codes.} Encrypting more than one bit requires encoding the bits to encrypt by using an efficient decodable code. We propose to use the ubiquitous Gabidulin codes~\cite{G85}. To properly define these  codes, let us first introduce  $q$-polynomials. They are polynomials of the form
$$
P(X) = p_0 X + p_1 X^q + \cdots + p_k X^{q^k}
$$
where $ p_i \in \Fqm $ and $ p_k \neq 0 $. The integer $ k $ is called the \emph{$q$-degree} of $ P $. Let $ \mathcal{L}_{<k} $ denote the set of~$q$-polynomials of~$q $-degree less than $ k $.

Gabidulin codes belong to a particular sub-class of
matrix codes:  $\mathbb{F}_{q^{m}}$--linear codes. Recall that an $\mathbb{F}_{q^{m}}$--linear code $\C$ with length $n$
and dimension $\kappa$ is a subspace with $\mathbb{F}_{q^m}$--dimension $\kappa$ of $\mathbb{F}_{q^{m}}^{n}$. We say that it has
parameters $[n, \kappa]_{q^m}$ or that it is an $[n, \kappa]_{q^m}$-code. It turns out that $\Fqm$--linear
codes are {\it isometric} 
to a particular subclass of matrix codes.  However,
to exhibit this isometry, we first need to define the underlying metric
for~$\Fqm$--linear codes. Given two vectors~$\vec{v}, \vec{w} \in \Fqm^{n}$, their rank-distance is
defined as
$$
|\vec{v} - \vec{w}| \eqdef \dim_{\Fq} \mathbf{Span}\left( v_1-w_1,\dots,v_n-w_n \right)\ .
$$
The rank weight of $\vec{v} \in \Fqm^{n}$
is denoted $|\vec{v}| \eqdef |\vec{v}-\vec{0}|$.
Notice that the aforementioned rank-weight~$|\vec{v}|$ is nothing but the rank of the matrix obtained by decomposing entries of
$\vec{v}$ in a fixed~$\Fq$--basis of $\Fqm$
viewed as an $\Fq$--vector space with dimension~$m$.  This
decomposition then gives us the aforementioned isometry. Furthermore, an $\mathbb{F}_{q^{m}}$--linear code with dimension $\kappa$ can be viewed as an $[m \times n, \kappa m]_q$-code after applying the isometry.

\begin{definition}[Gabidulin codes] Let $ m, n, \kappa$ be integers such that $\kappa\leq n \leq m $, and let $ \vec{g} = (g_1, \dots, g_n) \in \mathbb{F}_{q^m}^n $ be a vector whose entries are $ \mathbb{F}_q $-linearly independent.  The {Gabidulin code} of evaluation vector $ \vec{g} $ is the following $[n, \kappa]_{\mathbb{F}_{q^{m}}}$-linear code
	\[
	\mathsf{Gab}(\vec{g}, k) := \left\{ \left(P(g_1), \dots, P(g_n)\right) ~|~ P \in \mathcal{L}_{<k} \right\} .
	\]
\end{definition}

Gabidulin codes 
benefit from an efficient decoding algorithm whose properties are summarized in the following proposition. Many algorithms allow to decode Gabidulin codes. Here, the choice of the algorithm does not matter regarding the security of the scheme, although most recent algorithms allow for a decoding that can be considered as fast (see \cite{SCZ+} for instance). 

\begin{proposition}[\cite{G85}]\label{prop:gabidulin}
		Given a Gabidulin code $\mathsf{Gab}(\vec{g},\kappa)$ with parameters $m,n,\kappa,q$, {\it i.e.,} given the knowledge of~$\vec{g}\in \Fqm^{n}$ and $\kappa$, there exists a deterministic algorithm $\mathsf{Decode}^{\mathsf{Gab}}$ running in $O(n^{2})$ operations in $\Fqm$ and such that given $\vec{y} \in \Fqm^{n}$, $\vec{g}$ and $\kappa$, 
	\begin{itemize}\setlength{\itemsep}{5pt}
		\item if $\vec{y} = \vec{c} + \vec{e}$ where $\vec{c} \in \mathsf{Gab}(\vec{g},\kappa)$ and $|\vec{e}| \leq \frac{n-\kappa}{2}$, it outputs $\vec{e}$,

		\item otherwise, it outputs $\bot$~. 
	\end{itemize}
	\end{proposition}

Given a Gabidulin code $\mathsf{Gab}(\vec{g},\kappa)$, we will denote by $\mathsf{Gab}_{\vec{g}}(\vec{m}) \in \mathsf{Gab}(\vec{g},\kappa)$ the encoding of the vector~$\vec{m} \in \Fqm^{\kappa}$ into the code $\mathsf{Gab}(\vec{g},\kappa)$. We will also interpret the Gabidulin code as a matrix code, {\it i.e.,} $\mathsf{Gab}_{\vec{g}}(\vec{m}) \in \Fq^{m \times n}$ instead of $\Fqm^{n}$ via the aforementioned isometry. In particular, $\mathsf{Decode}^{\mathsf{Gab}}$ algorithm also applies to $m \times n$ matrix code arising from a Gabidulin code. 
\medskip

{\bf \noindent Vectorization of matrices.} For the sake of conciseness, we will sometimes consider matrices as vectors by using the bijection 
\begin{equation}\label{eq:rho} 
	\rho \text{ }: \text{ }\vec{A} = \begin{pmatrix}
		a_{1,1} & \dots & a_{1,n} \\
		\vdots & \ddots & \vdots \\
		a_{m,1} & \dots & a_{m,n}
	\end{pmatrix} \mapsto \rho(\vec{A}) \eqdef (a_{1,1} \dots a_{1,n} \dots a_{m,1}\dots a_{m,n}) \ .
\end{equation} 
The benefits of this consideration are three-fold:
\begin{itemize}\setlength{\itemsep}{5pt}
	\item[$\bullet$] It allows us to write the inner-product of matrices as the canonical inner-product for vectors, {\it i.e.,} $\langle \vec{A},\vec{B} \rangle = \tr(\vec{A}\vec{B}^\top) = \rho(\vec{A})\cdot \rho(\vec{B}) \eqdef \rho(\vec{A})\rho(\vec{B})^\top \ ;$
	\item[$\bullet$] It allows us to write a $\minrank$ instance $\vec{Y} + \sum\limits_{i=1}^{k} x_i\vec{B}_i = \vec{E}$ as 
	$$
	\vec{y} = \vec{xG} + \vec{e} 
	$$
	where $\vec{x}\in \Fq^{k}, \vec{y}  \eqdef \rho(\vec{Y}), \vec{e} \eqdef \rho(\vec{E})$ and $\vec{G} \in \Fq^{k \times mn}$ is the matrix such that its rows are all the~$\rho(\vec{B}_i)$'s, {\it i.e.} $\vec{G}^{\top} \eqdef \begin{bmatrix}
		\rho(\vec{B}_1)^{\top} & \cdots\;  & 
		\rho(\vec{B}_k)^{\top}
	\end{bmatrix}$;
	\item[$\bullet$] It allows us to view the matrix composed of the $\langle \vec{F}^{(i)}, \vec{E}^{(j)}\rangle$'s as the product $\vec{E}\vec{F}$ where $\vec{E}$ ({\it resp.} $\vec{F}$) is the matrix composed of the vectorized of $\vec{E}^{(j)}$'s ({\it resp.} $\vec{F}^{(i)}$'s).
\end{itemize} 
\medskip

{\bf \noindent Using a single matrix code as a public-key.} One of the limitations of the scheme described in Section \ref{sec:Encstructure} is the use of $\ell_1$ matrix codes to build the public-key. In fact, having this many codes implies using codes of small dimensions, making it harder to find suitable parameters. We will thus consider an instantiation of the scheme where all the codes are the same, {\it i.e.,} we take one bigger code. Only the public-key changes, we replace 
$$
\left( \left( \vec{B}_{\ell}^{(1)}\right)_{\ell \in [1,k_1]}, \vec{Y}^{(1)} \right), \dots,  \left( \left( \vec{B}_{\ell}^{(\ell_1)}\right)_{\ell \in [1,k_{\ell_1}]}, \vec{Y}^{(\ell_1)} \right)
$$ 
being the original public-key by: 
$$
\pk =  \left( \left( \vec{B}_{\ell}\right)_{\ell \in [1,k]}, \vec{Y}^{(1)},\dots,\vec{Y}^{(\ell_1)} \right) .
$$

{\bf \noindent Overview of the instantiation.} Our public-key now consists now in~$k$ matrices $\vec{B}_1,\dots,\vec{B}_k \in \Fq^{m \times n}$, and $\ell_1$ matrices $\vec{Y}^{(1)},\dots,\vec{Y}^{(\ell_1)} \in \Fq^{m \times n}$, such that 
$$
\forall j \in [1,\ell_1], \; \vec{Y}^{(j)} = \sum\limits_{\ell=1}^{k}x_\ell^{(j)}\vec{B}_\ell + \vec{E}^{(j)}
$$ 
for some vectors $\vec{x}^{(j)} \in \Fq^k$ and the $\vec{E}^{(j)}$'s being of rank less or equal to $r$ but with the same column support. Thus, thanks to the bijection $\rho$ (see Equation~\eqref{eq:rho}), it can be described as 
$$
\vec{T} \eqdef \vec{SG}+\vec{E} \in \Fq^{\ell_1 \times mn},
$$ 
where each row of $\vec{T}$ corresponds to a $\minrank$ instance and where $\vec{S} \in \Fq^{\ell_1 \times k}$ is the matrix composed of all the $x_\ell^{(j)}$'s.

In Section \ref{sec:Encstructure}, to encrypt one bit, we transmitted noisy codewords ${\vec{C}_j}^\perp + \vec{F}^{(j)}$'s where all the~${\vec{C}_j}^\perp$'s were in the dual of the matrix code $\mathbf{Span}\left( \C,\vec{E}^{(1)},\dots,\vec{E}^{(\ell_1)}\right)$ where~$\C$ was the sum of all the codes defined by the public-key. Here, we replaced these codes  by a single code.

To show how we will now encrypt bits we first change the notation to the vectorized one, and we will use the dual representation of $\minrank$. Let $\vec{H} \in \Fq^{mn-k-\ell_1 \times mn}$ be the matrix obtained via the vectorization of a basis of 
$$
\mathbf{Span}\left( \C,  \vec{E}^{(1)},\dots,\vec{E}^{(\ell_1)}\right)^{\perp} .
$$
Let $\vec{M}$ be uniformly sampled  in~$\Fq^{\ell_2 \times (mn-k-\ell_1)}$. Ciphertexts can now be written as
\begin{align*}
	\vec{Z} \eqdef \left\{ \begin{array}{ll} 
		\vec{U} ~~ &\text{If $b=1$} \\
		\vec{M}\vec{H}+ \vec{F} ~~ &\text{If $b=0$} 
	\end{array}\right.
\end{align*} where $\vec{F} \in \Fq^{\ell_2 \times mn}$ is composed of the vectorized $\vec{F}^{(i)}$'s (according to notation from Figure~\ref{algo:scheme}) and~$\vec{U}$ is a uniform $\ell_2 \times mn$ matrix.

Now the decryption consists as computing $\vec{E}\vec{Z}^{\top}$. When $b=0$ has been encrypted, it consists in
$$
\vec{E}\vec{Z}^\top = \vec{E}\vec{H}^\top\vec{M}+\vec{EF}^\top = \vec{EF}^\top
$$ 
as $\vec{EH}^\top = \vec{0}$. Indeed, its coefficients are given by inner product between the rows of $\vec{E}$ and rows of $\vec{H}$. The coefficients of $\vec{EF}^\top$ are then the $\langle \vec{F}^{(i)}, \vec{E}^{(j)}\rangle$'s (due to the correspondence between inner product of matrices and inner product of vectors as described above) and thus the decryption works as previously by checking the rank.

We can go even further, instead of transmitting $\ell_2$ matrices, one can transmit only inner products by using the $\minrank$ dual form. The ciphertext becomes

\begin{align*}
	\vec{Z} \eqdef
	 \left\{ \begin{array}{ll} 
		\vec{F}\begin{bmatrix}\vec{G} \\ \vec{T} \end{bmatrix}^\top  &\text{If $b=1$}\\
		\vec{U} \begin{bmatrix}\vec{G} \\ \vec{T} \end{bmatrix}^\top &\text{If $b=0$}
	\end{array}\right.
\end{align*} where $\vec{F}$ and $\vec{U}$ are the same as explained above. It is readily seen that $\vec{H}\begin{bmatrix}\vec{G} \\ \vec{T} \end{bmatrix}^\top  = \vec{0}$, so all we did was indeed to take a dual formulation.

Using this last formulation for the ciphertext, encrypting a message $\vec{m} \in \mathbb{F}_{q^{\ell_1}}^{\kappa}$ becomes clearer: the message should be encoded in the Gabidulin code and then added to $\vec{FT}^\top$. The ciphertext then becomes 
$$
\begin{bmatrix}
	\vec{U} & \vec{V}
\end{bmatrix} \eqdef \vec{F}\begin{bmatrix}\vec{G} \\ \vec{T} \end{bmatrix}^\top +\begin{bmatrix}
	\vec{0}^{\ell_2 \times K} &  \mathsf{Gab}_{\vec{g}}(\vec{m})
\end{bmatrix}.
$$ This allows to decrypt by computing 
$$
\vec{V} - \vec{US}^\top = \vec{FG}^\top \vec{S}^\top+ \vec{FE}^\top + \mathsf{Gab}_{\vec{g}}(\vec{m}) - \vec{FG}^\top\vec{S}^\top =  \vec{FE}^\top + \mathsf{Gab}_{\vec{g}}(\vec{m}),
$$ 
and then decode it thanks to the Gabidulin decoder. 
We fully describe our scheme, that we call \textsf{MinRankPKE}, in Figure \ref{alg:MinRankPKE}.

\begin{figure}[h!]
	\begin{tcolorbox}[colframe=black, colback=white, boxrule=0.5mm, sharp corners]
		\begin{itemize}\setlength{\itemsep}{5pt}
			\item $\mathsf{Parameter}(\lambda)$: $m,n,\ell_1,\ell_2, k,r,d,q$ be integers such that 
			$$
			k \leq mn \quad \mbox{and} \quad rd < \min(\ell_1,\ell_2) \ .
			$$

			\item$\mathsf{Keygen}(m,n,\ell_1, k, r, q)$: 
			
			$E \Unif \left\{ X \subseteq \mathbb{F}_{q}^{m} \mbox{ with dimension } d \right\}$ 
			
			\texttt{for $j \in [1,\ell_1]$:}
			
			\quad $\vec{E}^{(j)} \Unif  \left\{ \vec{E} \in \mathbb{F}_{q}^{m \times n}: \; \Sp{\vec{E}^{}} =E \right\}$ 
			
			$\vec{S} \Unif \Fq^{\ell_2 \times k}$
			
			$\vec{G} \Unif \mathbb{F}_{q}^{k \times mn}$ 
			
			\texttt{sample} $\mathsf{Gab}(\vec{g},\kappa)$ an $[\ell_2,\kappa]_{q^{m}}$ Gabidulin code on $\mathbb{F}_{q^{\ell_1}}$, 
			
			\texttt{compute} $\vec{E}^{\top} \eqdef \begin{bmatrix}
				\rho(\vec{E}^{(1)})^{\top}& \cdots & 
				\rho(\vec{E}^{(\ell_1)})^{\top}
			\end{bmatrix}$ 
			
			$\mathsf{sk} \leftarrow  \vec{S}$
			
			$\mathsf{pk} \leftarrow \vec{T} \eqdef \vec{SG}+\vec{E}$, $\mathsf{Gab}(\vec{g},\kappa)$
		\end{itemize}
		\medskip

		\begin{center}
			\renewcommand{\arraystretch}{1.3}
			\begin{tabular}{l@{\hspace{3mm}}|@{\hspace{3mm}}l}
				$\mathsf{Enc}^{\mathsf{pk}}(\vec{m})\!\!: \qquad \qquad \qquad$ & $\mathsf{Dec}^{\mathsf{sk}}\left(\vec{Y}^{(1)},\dots,\vec{Y}^{(\ell_2)}\right)\!\!:$ \\
				$F \Unif \left\{ X \subseteq \mathbb{F}_{q}^{n} \mbox{ with dimension } d \right\}$ & \texttt{compute} $\vec{W} \eqdef \vec{V}-\vec{US}^\top$ \\
				\texttt{for} $i \in [1,\ell_2]$: & \texttt{compute} $\Tilde{\vec{m}} \eqdef \mathsf{Decode}^{\mathsf{Gab}}(\vec{W})$\\
				\quad $\vec{F}^{(i)} \Unif  \left\{ \vec{F} \in \mathbb{F}_{q}^{m \times n}: \; \Sp{\vec{F}^{\top}} =F \right\}$ & \texttt{return} $\Tilde{\vec{m}}$. \\
				 \texttt{compute }$\vec{F}^{\top} = \begin{bmatrix}
					\rho(\vec{F}^{(1)})^{\top} & \cdots 
					\rho(\vec{F}^{(\ell_2)})^{\top}
				\end{bmatrix}$ & \\
				 \texttt{compute } $\vec{U} = \vec{F}\vec{G}^\top$ & \\
				  \texttt{compute } $\vec{V} = \vec{F}\vec{T}^\top + \mathsf{Gab}_{\vec{g}}(\vec{m})$ & \\
				\texttt{return} $(\vec{U},\vec{V})$. & \\
			\end{tabular} 
\end{center} 
	\end{tcolorbox}
	\caption{The $\mathrm{MinRankPKE}$ encryption scheme}
	\label{alg:MinRankPKE}
\end{figure} 

\begin{comment}
	\begin{algorithm}[H]
		\caption{The $\mathrm{MinRankPKE}$ encryption scheme}
		\label{alg:MinRankPKE}
		\begin{algorithmic}[1]
			\Keygen Sample $\vec{G} \in \Fq^{K \times mn}$,  $\vec{S} \in \Fq^{\ell_1 \times K}$ uniformly at random, and a $[\ell_2,k]$ Gabidulin code on $\mathbb{F}_{q^{\ell_1}}$, $\mathsf{Gab}$. Sample a subspace $E \subset \Fq^m$ of dimension $r$, and $\ell_1$ matrices $\vec{E}^{(i)} \in E^n$. Compute $\vec{E} \in \Fq^{\ell_1 \times mn}$ composed of all the $\rho(\vec{E}^{(i)})$. The public key is $\pk = (\mathsf{Gab}(\vec{g},k),\vec{G},\vec{SG}+\vec{E})$, the private key is $\sk = (\vec{S},\vec{E})$.  
			\Enc To encrypt $\vec{m} \in \mathbb{F}_{q^{\ell_1}}^{k}$, sample a subspace $F \subset \Fq^n$ of dimension $d$, and $\ell_2$ matrices $\vec{F}^{(i)} \in F^m$. Compute $\vec{F} \in \Fq^{\ell_2 \times mn}$ composed of all the $\rho({\vec{F}^{(i)}}^\top)$.  Compute $\vec{U} = \vec{F}\vec{G}^\top$ and $\vec{V} = \vec{F}\vec{T}^\top + \mathsf{Gab}_{\vec{g}}(\vec{m})$. Output $(\vec{U},\vec{V})$. 
			\Dec Compute $\vec{W}= \vec{V}-\vec{U}\vec{S}^\top$. Compute $\Tilde{\vec{m}} = \mathsf{Decode}_{\mathsf{Gab}}(\vec{W})$ and return $\Tilde{\vec{m}}$.
		\end{algorithmic} 
	\end{algorithm}
\end{comment}
\begin{proposition}[Correctness of the decryption]
	If $rd$ is smaller than $\lfloor \frac{\ell_2-\kappa}{2}\rfloor$, the $\mathsf{Dec}^{\sk}$ algorithm from Figure \ref{alg:MinRankPKE} returns the message $\vec{m}$.
\end{proposition}
\begin{proof}
	We have the following computation,
	\begin{align*}
		\vec{V}-\vec{U}\vec{S}^\top&=\vec{F}\vec{T}^\top+ \mathsf{Gab}_{\vec{g}}(\vec{m}) - \vec{F}\vec{G}^\top\vec{S}^\top\\ &= \vec{F}\vec{G}^\top\vec{S}^\top+\vec{F}\vec{E}^\top-\vec{F}\vec{G}^\top\vec{S}^\top+\mathsf{Gab}_{\vec{g}}(\vec{m})\\ &= \vec{F}\vec{E}^\top +\mathsf{Gab}_{\vec{g}}(\vec{m})\in \Fq^{\ell_1 \times \ell_2}
	\end{align*}
	Thanks to the correspondence between the two formulations of the scalar product, one can see that $\vec{F}\vec{E}^\top$ is the matrix $\vec{D}$ such that $\vec{D}(i,j) = \langle \vec{F}^{(i)},\vec{E}^{(j)} \rangle$. Thus, the rank of $\vec{F}\vec{E}^\top$ is a straightforward application of Theorem \ref{theo:fund} as in the proof of Proposition \ref{prop:gabidulin} in which case the codeword can be decoded. 
\end{proof}

{\bf \noindent Security of the instantiation.} Theorem \ref{theo:secu} can easily be adapted to the instantiation we just described (recall that adding a decodable code in the construction does not affect the security). If an attacker can break this version of the scheme with a public-key containing one code, then we can use this attacker to break the original instantiation with many codes in the public-key and thus solve stationary-$\minrank$ via the search-to-decision reduction. In what follows, we will be interested in finding the best attacks against this new version of the scheme to provide parameters. It is why we will actually be interested in solving the following variation of stationary-$\minrank$: $\MinRank$ Support Learning (\textsf{MSL}). 
Providing parameters based on the best stationary-$\minrank$ solvers would also give a reliable instantiation, but the parameters would not be tight at all, given the losses in reduction from multiple matrix codes in the public key to only one code.

\begin{definition}[$\MinRank$ Support Learning ($\MSL$)] \label{def:msl}
	Let~$m,n, N, k,t,q$ be integers that are functions of some security
	parameter $\lambda$ and such that $mn \geq k$. Let
	$\vec{B}_1, \dots, \vec{B}_{k} \in \mathbb{F}_{q}^{m \times
		n}$ and~$\lambda^{(j)}_1,\dots,\lambda^{(j)}_{k} \in \mathbb{F}_{q}$ for $j \in [1,N]$ be sampled uniformly at
	random. Let $E \subseteq \mathbb{F}_{q}^{m}$ be a random subspace with dimension~$t$ and~$\vec{E}^{(j)} \in \left\{ \vec{X} \in \mathbb{F}_{q}^{m \times n}: \; \Sp{\vec{X}}= E \right\}$ be uniformly distributed  for~$j \in [1,N]$. Let,
	$$
	\forall j \in [1,N], \; \vec{Y}^{(j)} \eqdef \sum_{\ell=1}^{k} \lambda^{(j)}_{\ell} \vec{B}_{\ell} + \vec{E}^{(j)} .
	$$
	The $\mathsf{MSL}(m,n,N,k,t,q)$ problem
	consists, given $\left( \vec{B}_{\ell}\right)_{\ell \in [1,k]},\left(\vec{Y}^{(j)}\right)_{j \in [1,N]}$, in finding $E$.
\end{definition}

It turns out \textsf{MSL} is harder than stationary-$\minrank$. Indeed, suppose that we are given an instance of stationary-$\minrank$: $\mathcal{C}_1,\dots,\mathcal{C}_N$ and $\vec{C}_1 + \vec{E}_1, \dots, \vec{C}_N + \vec{E}_N$ where the $\mathcal{C}_{i}$'s are random matrix codes with dimension $k_i$, the $\vec{C}_{i}$'s are random codewords in the $\mathcal{C}_{i}$'s and the $\vec{E}_{i}$'s all have the same column support. Then,
	$
	\mathcal{C} \eqdef \sum_{i} \mathcal{C}_{i}
	$
	is a random code with dimension $\sum_{i} k_i$ (with overwhelming probability under the condition that $\sum_{i} k_i < mn$). Furthermore, let $\vec{C}'_1, \dots, \vec{C}'_{N}$ be picked uniformly at random in $\mathcal{C}$. It is easily verified that the $\vec{Y}'_{i} \eqdef \vec{C}'_i + \vec{Y}_i$'s together with $\C$ form a valid (average) instance of \textsf{MSL}. 

\section{Algorithmic hardness of $\minrank$ Support Learning and stationary-$\minrank$} \label{sec:security}

This section is devoted to studying the hardness of the problem upon which the concrete security of our encryption scheme is based on: $\mathsf{MSL}$. Furthermore, we will also study the hardness of stationary-$\minrank$, to which we reduce security via a search-to-decision reduction. Our proposed algorithms to solve these two problems will be treated independently in what follows.

First, we will focus on $\MSL$. To derive algorithms solving it, we will draw inspiration from the best algorithms to solve the  well-known Rank Support Learning ($\RSL$) problem, which corresponds to $\MSL$ where an additional $\mathbb{F}_{q^{m}}$--linear structure is added. Our proposed algorithms to tackle $\MSL$ can then be interpreted as an adaptation of~\cite{GHPT17,BB21,BBBG24} where the $\mathbb{F}_{q^{m}}$--linearity can no longer be used. 
Then, we will discuss the hardness of stationary-$\minrank$. As we will see, the best algorithms we found are all derived from $\minrank$-solvers. 
\medskip

{\bf \noindent On the average number of solutions of an $\MSL$ instance.} In \cite{Dyseryn25}, the average number of solutions of a random $\RSL$ instance was given. This is important as it influences the complexity of the attack. Here, we can easily adapt the formula. The number of solutions for a $\MSL$ instance is given by 
$$
\qbinom{m}{t}\cdot \frac{q^{tnN}}{q^{N(mn-k)}}.
$$
To explain this formula, we can proceed recursively. Taking the first instance, we know that there are on average 
$$
\qbinom{m}{t}\cdot \frac{q^{tn}}{q^{mn-k}}
$$ 
solutions. Then, the probability that the second instance also possesses a solution with the same support is 
$$
\frac{q^{tn}}{q^{mn-k}},
$$ 
and so on until the $N$th instance, thus the number of solutions. In our case, it is always lower than $1$.

\subsection{A general approach for the $\MSL$ problem}\label{subsec:GenApp}

Our approach to solve $\MSL$ is analogous to the one proposed in \cite{GHPT17,BB21,BBBG24} treating $\RSL$.
It corresponds in our case to proceed as follows. First, we build a larger decoding problem with many solutions. Then, we can deduce the solution of the $\MSL$ problem by solving a $\minrank$ instance. The advantage of this approach is that
the considered $\minrank$ instance has its parameters reduced, due to the $q^N$ solutions in the large instance.
\medskip 

{\bf \noindent Building a larger code.} In the $\MSL$ problem, we are given $N$ instances of $\MinRank$ are which are obtained using the same code, and the same column support for the different errors. We note the $j$th instance as ${\vec{Y}}^{(j)} = \sum\limits_{i=1}^{k}x_{i}^{(j)}\vec{B}_i + \vec{E}^{(j)}$ with $|\vec{E}^{(j)}| \le t$. Then, we proceed as follows:
\medskip 
\begin{itemize}\setlength{\itemsep}{5pt}
	\item[$\bullet$] Build a code 
	$$
	\C_{aug} \eqdef \mathbf{Span}\left( \vec{B}_1,\dots,\vec{B}_k,{\vec{Y}}^{(1)},\dots,{\vec{Y}}^{(N)} \right)
	$$ 
	which has dimension $k+N$;
	\item[$\bullet$] Find one of the $q^N$ codewords in $\C_{aug}$ that is in 
	$$
	\C' \eqdef \mathbf{Span}\left( \vec{E}^{(1)},\dots, \vec{E}^{(N)}\right). 
	$$
	This is because the $\vec{E}^{(j)}$'s belong to $\C_{aug}$.
\end{itemize}
\medskip 
The following lemma \ref{lemma:dimCaug} is well-known (\cite{GHPT17}, \cite{BB21}).
\begin{lemma}\label{lemma:dimCaug}
	Let $\C' =\mathbf{Span}\left(  \vec{E}^{(1)},\dots, \vec{E}^{(N)} \right)$. Then for all $\vec{E} \in \C'$, $|\vec{E}| \le t$, $\C' \subseteq \C_{aug}$ and $\dim(\C_{aug}) \leq k+N$.
\end{lemma}

A direct implication of this lemma is that it is possible to solve the $\MSL$ problem by finding one of the $q^N$ codewords of rank $t$ in the code $\C_{aug}$, {\it i.e.,} a $\MinRank$ instance.
\medskip

{\bf \noindent Reducing the number of solutions.} To solve this bigger instance, it makes sense to reduce the number of solutions, as the $q^N$ solutions directly give several ways to reduce the parameters. This is done either by finding a matrix of rank \textit{strictly less} than $t$, or by specializing some columns of the solution to columns of zeroes. Both approaches should be considered. 
\medskip

{\bf \noindent An error of smaller rank.} We begin with the matrix of smaller rank, by following \cite{GHPT17} and \cite[Proposition 1]{BB21}.

\begin{lemma}[\cite{BB21}] \label{lemma:cardC'}
	The expected number of codewords of rank $w$ in $\C'$ is 
	$$
\frac{\mathcal{S}_{t,n,w}}{q^{tn-N}}\ .
	$$
\end{lemma}
\begin{proof}
	We give the proof in an informal way here, as it is exactly as in \cite[Proposition 1]{BB21}: the rank of codewords in $\C'$ is determined by the rank of the linear combinations of the $\vec{P}^{(i)}$ where~$\vec{E}^{(i)} = \vec{V}\vec{P}^{(i)}$ and $\vec{P}^{(i)} \in \Fq^{t \times n}$. As there are $N$ such $\vec{P}^{(i)}$'s, they generate a $[t\times n,N]_q$ matrix code. The density of matrices of rank $w$ in $\Fq^{t \times n}$ is given by $\frac{\mathcal{S}_{t,n,w}}{q^{tn}}$, which has to be multiplied by $q^{N}$, thus the result. 
\end{proof}

\begin{lemma} \label{lemma:N/n}
	Assuming $ \vec{E}^{(1)},\dots, \vec{E}^{(N)} $ are linearly independent, one can expect to have a codeword of rank $t-\delta$ for all $\delta$ such that $N \ge \delta(n-t+\delta)$. 
\end{lemma}
\begin{proof}
	It is the same as \cite{BB21}: the number $\mathcal{S}_{t,n,w}$ is approximated by $q^{w(t+n-w)}$ when $q \rightarrow \infty$, thus the result when $w = t-\delta$. 
\end{proof}

In the same fashion as $\RSL$, this lemma shows that there is a very simple way to have a gain in complexity when $N$ grows. In fact, Corollary \ref{cor:reductionN/n} shows that a $\MinRank(m,n,N,K+N,t,q)$ instance coming from $\MSL$ reduces into a $\minrank(m,n,K+N,t-\delta,q)$ one. However, this does not indicate any gain whenever $N < \delta(n-t+\delta)$.

\begin{corollary}\label{cor:reductionN/n}
	We can solve $\MSL(m,n,N,k,t,q)$ by solving $\minrank(m,n,k+N,t-\delta,q)$ where~$N \ge \delta(n-t+\delta)$.
\end{corollary}

\begin{proof}
	We know that there is a linear combination of the errors that is of rank $t-\delta$. Thus, there is an error that has the same support as $\vec{E}^{(1)},\dots,\vec{E}^{(N)}$, but has rank $t-\delta$. This then corresponds to a $(m,n,k+N,t-\delta,q)$ $\MinRank$ instance. 
\end{proof}

{\bf \noindent Shortening the code.} It is possible to perform another specialization, which leads to a better gain in the complexity, especially in the regime $t\le N-1$. This consists of specializing the columns to zeroes instead of reducing the rank, a process called {\em shortening}.

\begin{lemma} \label{lemma:N/t}
	Assuming $ \vec{E}^{(1)},\dots, \vec{E}^{(N)} $ are linearly independent, there is always at least one linear combination $\sum\limits_{i=1}^{N}\lambda_i \vec{E}^{(i)} = \begin{pmatrix}
		\vec{0}^{m \times a} & \Tilde{\vec{E}}
	\end{pmatrix}$ where $a= \lfloor (N-1)/t\rfloor$.
\end{lemma}
\begin{proof}
	It is the same as \cite{BB21} and \cite{BBBG24}. We know that $\vec{E}^{(i)} = \vec{V}\vec{P}^{(i)}$ for a fixed $\vec{V} \in \Fq^{m \times t}$ and $\vec{P}^{(i)} \in \Fq^{t \times n}$ with $\vec{V}$ and each $\vec{P}^{(i)}$ of rank $t$. Then, we know by definition that there is a linear combination $\sum\limits_{i=1}^{N}\lambda_i \vec{P}^{(i)}$ that has $a = \lfloor (N-1)/t \rfloor$ columns equal to zero. The result follows immediately. 
\end{proof}

A consequence of this lemma is that it makes it possible to reduce the $(m,n,k+N,t,q)$ instance into a $(m,n-a,k+N-am,t,q)$ one.

\begin{corollary}\label{cor:reduction}
	We can solve $\MSL(m,n,N,K,t,q)$ by solving a $\minrank(m,n-a,K+N-am,t,q)$ where $a =  \lfloor (N-1)/t \rfloor$.
\end{corollary}
\begin{proof}
	We know that there is a linear combination of the errors with $a$ columns that are zeroes. We are thus in the situation where we have that, for some values of $x_1,\dots,x_k,\lambda_1,\dots,\lambda_N$, \begin{equation} \label{eq:spec0} 
		\sum\limits_{i=1}^{k} x_{i}\vec{B}_i+\sum\limits_{i=1}^{N} \lambda_{i}{\vec{Y}}^{(i)} = \vec{E}= \begin{pmatrix}
			\vec{0}^{m \times a} & \widetilde{\vec{E}}
		\end{pmatrix}. 
	\end{equation} Using the $am$ linear equations that correspond to the $a$ first columns, The set of variables is then reduced to $\mu_1,\dots,\mu_{k+N-ma}$, such that $\vec{E} = \sum\limits_{i=1}^{k+N-ma} \mu_i\vec{B}'_i$ where each $\mu_i$ correspond to a linear combination of the $x_i$'s and $\lambda_j$'s, and where each $\vec{B}'_i$ is a linear combination of the $\vec{B}_j$'s and ${\vec{Y}}^{(\ell)}$'s, for all $ i\in\oneto{k+N-ma}$, $j \in \oneto{k}$, $\ell \in \oneto{N}$. To solve the $\MSL$ instance, one can then shorten the code by removing the first $a$ columns. It corresponds to a $\MinRank(m,n-a,k+N-am,t,q)$ instance.
\end{proof}

\begin{remark}
	This reduction of parameters corresponds exactly to what is actually done in the hybrid approach on $\MinRank$ and $\RSD$ \cite{BBB+23}. However, this reduction is done ``for free" here, due to the number of solutions.
\end{remark}

{\bf \noindent Combining the two approaches.} These two approaches are not mutually exclusive, and it is beneficial to combine them. If $N-\delta (n-t+\delta)>0$, it is still possible to try and apply the second approach on top of it. In fact, there are approximately $q^{N-\delta(n - t+\delta)}$ codewords of rank $t-\delta$. Thus, one shortens the code by taking the same method as previously, but with~$a = \lfloor\frac{ N-1-\delta(n - t+\delta)}{(t-\delta)}\rfloor$ instead. There are then approximately $q^{N - \delta(n - t+\delta) - a (t-\delta)}$ codewords of rank $t-\delta$ in the shortened code. This multiplicity of solutions allows us to reduce the dimension of the code once more, by~$N - \delta(n - t+\delta) - a (t-\delta)$. The complexity will be taken by using the optimal values of $\delta$ and~$a$.

Essentially, what this combination tells us is that it is possible to do trade-offs between the rank of the word we are looking for and the shortening we consider. Taking $\delta = 0$ obviously comes down to only shortening the code. Doing so, we know there will be $q^{N-at}$ correct codewords and so the code can have a dimension reduced by $N-at$. When $\delta >0$, the same thing appears, leading to our previous explanation.
\medskip

{\bf \noindent Hybrid approach.} As this will apply to all the attacks used to solve a $\MinRank$ instance, we briefly recall the complexity of the hybrid approach given by \cite{BBB+23} and in \cite{MiRitH}. This approach consists in multiplying the matrix code by a matrix $\widetilde{\vec{P}}$, and making the bet that this multiplication makes the first $a$ columns of the error to be $\vec{0}$, allowing a reduction of parameters as previously explained. Several $x_i$'s can also be guessed, to further reduce the dimension \cite{courtois,FSS10,BBC+20}. Thus, for a cost of $q^{\ell t+v}$ repetitions, it is possible to reduce a $\MinRank$ instance of parameters $(m,n,k,t,q)$ into an instance with parameters $(m,n-\ell,k-\ell m-v,t,q)$. Hence, the running-time of this approach is given by,
$$
O\left(q^{\ell t+v}\mathsf{C}_{\MinRank}(m,n-\ell,k-\ell m-v,t,q) \right)
$$ for optimal values $\ell$ and $v$. 

\subsection{Combinatorial attacks} \label{subsec:combinatorialattacks}

We recall the well-known kernel attack in Algorithm \ref{alg:kernelattack}. When there is only one solution in  $\MinRank(m,n,k+N,t,q)$, its complexity is given by
$$
\mathsf{C}_{\mathrm{Kernel}}(q,m,n,k,t) = O\left(k^\omega q^{t\lceil \frac{k+N}{m} \rceil}\right).
$$

\begin{algorithm}[h]
	\caption{Kernel attack on a MinRank instance with parameters $(m,n,k+N,t,q)$}
	\label{alg:kernelattack}
	\begin{algorithmic}[1]
		\Require Matrices $\vec{B}_1,\dots,\vec{B}_{k+N} \in \Fq^{m\times n}$.
		\Ensure A non-zero matrix $\vec{E} \in \mathbf{Span}\left( \vec{B}_1,\dots,\vec{B}_{k+N} \right)$ of rank equal to or smaller than $t$.
		\State Set $\ell = \lceil \frac{k+N}{m} \rceil$
		\Repeat 
		\State Sample a space $W \subseteq\Fq^{n}$ of dimension $\ell$, with matrix representation $\vec{W} \in \Fq^{n \times \ell}$.
		\State Set ${\vec{Y}} = \sum\limits_{i=1}^{k+N}x_i\vec{B}_i$ with unknowns $x_1,\dots,x_{k+N}$.
		\State Solve the linear system ${\vec{Y}}\vec{W} = \vec{0}$ in the $\{x_i\}_{i \in \oneto{k+N}}$ and compute the matrix~$\vec{E}$ associated to the solution.
		\Until{$|\vec{E}| \le t$}
		\State\Return $\vec{E}$.
	\end{algorithmic} 
\end{algorithm}

\begin{remark}
	One should notice that the algorithm is very similar to the one in \cite[Section 4.3]{GHPT17}. The reason is very simple: they are actually the same algorithms. The only difference is that the kernel attack aims at guessing a space of dimension~$\lceil (k+N)/m \rceil$ that lies in the kernel of the matrix $\vec{E}$, while the algorithm from \cite{GHPT17} aims at guessing a space of dimension~$\lceil m-(k+N)/n\rceil$ in which the error lies. The two algorithms thus perform exactly the same thing (when considering the transpose of the code). In \cite[Theorem 2]{GHPT17}, the behaviour of this algorithm in such conditions is already analyzed. Its running-time is given in \cite[Theorem 2]{GHPT17}. The running-time of the Kernel attack on the $\MinRank$ instance built from $\C_{aug}$, when not considering any shortening is given by 
	$$ \widetilde{O}\left(q^{\min(e_-,e_+)} \right)
	$$ where $e_- = (t-\delta)(\ell-\delta)$ and $e_+ = (t-\delta-1)(l-\delta-1)+n(\delta+1)-N$, with $\ell=  \lceil (k+N)/m \rceil$ and $\delta = \lfloor N/n \rfloor$.
\end{remark}

The formula we just described is not optimal however: we will always shorten the code as it is always beneficial. This leads to the following proposition.

\begin{proposition} \label{prop:compkernel}
	There is a combinatorial algorithm that solves the $\MSL(q,m,n,k,t,N)$ with running-time
	\begin{align*}
		{O}\Big(&\mathsf{C}_{\mathrm{Kernel}}(q, m, n-a, k -a m + \delta(n - t+\delta) + a (t-\delta), t-\delta)\Big)
	\end{align*}
	where $\delta$ and $a'$ are such that $N> \delta(n-t+\delta)+a'(t-\delta)$.
\end{proposition}
\begin{proof}
	This is a straightforward application of Corollaries \ref{cor:reductionN/n} and \ref{cor:reduction}, the kernel algorithm is applied to a shortened code with a reduced dimension, hence the result.  
\end{proof}

\begin{remark}
	The hybrid approach from Section \ref{subsec:GenApp} still applies, which can reduce the complexity even further. This should be taken into account: the whole complexity is actually 
	$$
	O\Big(q^{\ell t+v}\mathsf{C}_{\mathrm{Kernel}}(q, m, n-a-\ell, k -a m + \delta(n - t+\delta) + a (t-\delta)-\ell m-v, t-\delta)\Big)
	$$
	where we optimize over $\ell$ and $v$.
\end{remark}

{\bf \noindent A bound for a polynomial complexity.} It is well-known that the $\RSL$ problem becomes solvable in polynomial time if enough instances of decoding are given. It is the same for $\MSL$. 

\begin{comment}
	
	\begin{proposition}\label{prop:polynomial}
		There exists a combinatorial attack on the {\bf \noindent MinRank Support Learning Problem} with complexity 
		$$ \Tilde{\mathcal{O}}\left( q^{r(m- \lfloor \frac{mn-K-N}{n-a}\rfloor)} \right)$$ where $a = \lfloor (N-1)/t \rfloor$. \end{proposition}
	
	\begin{proof}
		While computing $ \Tilde{\mathcal{O}}\big(\mathsf{C}_{\mathrm{Kernel}}(q, m, n - a, K + N - am, r)\big)$ does not give the same complexity  (which is the complexity of the Kernel Attack when considering only the shortening, i.e., $\delta = 0$), we recall that the kernel algorithm can apply on the transpose of the code. We thus compute the complexity $$\Tilde{\mathcal{O}}\big(\mathsf{C}_{\mathrm{Kernel}}(q, n - a, m, K + N - am, r)\big) = q^{r(\frac{K+N-am}{n-a})}$$ where $a = \lfloor (N-1)/t \rfloor$. A simple computation shows that $m-\frac{mn-K-N}{n-a} = \frac{K+N-am}{n-a}$, hence the result.
	\end{proof}
	\begin{remark}
		Note that we ignore here the hybrid approach and consider only the shortening of the code to keep computations simple and to obtain the same bound as \cite{BBBG24}. The same approach could be done with other specializations, hybrid approaches, and without considering the transpose of the code.
	\end{remark}
	
\end{comment}

\begin{corollary}
	If $N \ge \frac{kt+m}{m-1}$, then $\MSL(q,m,n,N,k,t)$ is solvable in polynomial time.
\end{corollary}

\begin{proof}
	We simply solve $0 = \lfloor \frac{k-am+N}{m}\rfloor$ to remove the exponential part of the complexity of $\mathsf{C}_{\mathrm{Kernel}}(q,m,n-a,k+N-am,t)$. By considering $a = (N-1)/t$ is an integer, we obtain the result in a straightforward manner. 
\end{proof}

\begin{remark}
	We omit many specializations such as rank reduction, hybrid approach, or even reducing the dimension by $N-at$. We do this to keep computations simpler. The obtained bound would be only slightly different and the result would be roughly the same: $\MSL$ becomes solvable in polynomial time whenever~$N$ is greater than approximately $kt/m$. 
\end{remark}

\subsection{Algebraic attacks} \label{subsec:algebraicattacks}

\subsubsection{Support Minors.} We now explain the Support Minors modelling from \cite{BBC+20,BBB+23,BG25}. The goal of this modelling is to obtain a large number of quadratic equations, and hope that it is possible to linearize the system. We explain the modelling when applied to $\C_{aug}$ without reductions.

Let $\vec{E}$ be the error of rank $t$ that must be found in the matrix code 
$$
\C_{aug} \eqdef \mathbf{Span}\left(\vec{B}_1,\dots,\vec{B}_k,{ \vec{Y}}^{(1)},\dots,{\vec{Y}}^{(N)}\right).
$$ 
First, let $\vec{S} \in \Fq^{m \times t}$ and $\vec{C} \in \Fq^{t \times n}$ such that $\vec{E} = \vec{SC} = \sum\limits_{i=1}^{k} x_{i}\vec{B}_i+\sum\limits_{i=1}^{N} \lambda_{i}{\vec{Y}}^{(i)}$. 

To obtain the quadratic equations, set $r_j$ the $j$th row of $ \sum\limits_{i=1}^{k} x_{i}\vec{B}_i+\sum\limits_{i=1}^{N} \lambda_{i}{\vec{Y}}^{(i)}$. Then, the matrix 
$$
\widetilde{\vec{C}} = \begin{pmatrix}
	r_j\\\vec{C}
\end{pmatrix}
$$ 
is of rank $t$ and its maximal minors are zeros. This results in Modelling \ref{mod:SMM}.

\begin{modeling}[$\MSL$ Support Minors Modeling] \label{mod:SMM}
	Let $\vec{C}$ be a matrix of unknowns of size $t \times n$. We consider the system given by the maximal minors of the matrices $\begin{pmatrix}
		r_j\\\vec{C}
	\end{pmatrix}$, i.e., $$\Big\{f=0 ~\vert~ f \in \normalfont{{\bf \noindent MaxMinors}}\begin{pmatrix}
		r_j\\\vec{C}
	\end{pmatrix},~~~j\in \oneto{m}\Big\}$$ This system is composed of:
	\medskip 
	\begin{itemize}\setlength{\itemsep}{5pt}
		\item[$\bullet$] $k+N+\binom{n}{t}$ variables $x_1,\dots,x_k,\lambda_1,\dots,\lambda_N$ and the $c_T$'s, with $T \subset \{1..n\}$, $|T| = t$, which represents maximal minors of $\vec{C}$;
		\item[$\bullet$] $m\binom{n}{t+1}$ bilinear equations with coefficients in $\Fq$.
	\end{itemize}
\end{modeling}

If linearization is not possible immediately, the quadratic equations are all multiplied by all the variables $x_i$ and $\lambda_i$, until there are enough equations compared to the number of monomials. We refer to \cite{BBC+20} for more explanations on their numbers, and summarize it in Heuristic \ref{heur:SMMnbeq}.

\begin{heuristic}\label{heur:SMMnbeq}
	The number of linearly independent equations of bi-degree (1,b) obtained from Modelling \ref{mod:SMM} is $$\mathcal{N}_{b}(m,n,k+N,t) =  \sum_{i = 1}^{b}(-1)^{i+1} \binom{n}{ t+i}\binom{k+N+b-1-i}{b-i}
	\binom{m+i-1}{i}$$
	
	The number of monomials that appear in these equations is then 
	$$
	\mathcal{M}_b(m,n,k+N,t)= \binom{k+N+b-1}{b}\binom{n}{t}
	$$
\end{heuristic}

After the reductions from Section \ref{subsec:GenApp}, these values are 
$$\mathcal{N}_{b}^{\mathrm{red}} = \mathcal{N}_{b}^{\mathrm{red}}(m,n-a-\ell,k-am+ \delta(n - t+\delta) + a (t-\delta)-\ell m,t-\delta)
$$ 
$$
\mathcal{M}_b^{\mathrm{red}} = \mathcal{M}_b^{\mathrm{red}}(m,n-a-\ell,k-am+ \delta(n - t+\delta) + a (t-\delta)-\ell m,t-\delta).$$

Over $\mathbb{F}_2$, it is beneficial to multiply by all monomials of degree lower than $b$. The complexity of solving the instance is then $\mathcal{O}\left( q^{\ell t+v}\mathcal{N}_{\le b}^{\mathrm{red}} \mathcal{M}_{\le b}^{\mathrm{red}}  \right)$ for the first value of $t+2 > b \ge 1$ such that~$\mathcal{N}_{\le b}^{\mathrm{red}}  \ge \mathcal{M}_{\le b}^{\mathrm{red}} -1$. 

\begin{remark}
	A reader should note that thanks to the use of $\MinRank$ and $\MSL$ instead of decoding random $\Fqm$--linear codes and $\RSL$, the modelling is already well-known and analyzed thoroughly, contrary to new modellings for $\RSL$ such as~\cite{BB21}.
\end{remark}

\subsubsection{Minors}\label{subsub:Minors}
The Support Minors modelling is not the only algebraic attack on $\MinRank$. In fact, the minors modelling previously existed and has been analyzed in \cite{FSS10,FSS13}. This algorithm must not be neglected, as it can perform better than Support Minors in some cases.

The Minors modelling consists simply in computing the matrix $\widetilde{\vec{E}}=  \sum\limits_{i=1}^{k} x_{i}\vec{B}_i+\sum\limits_{i=1}^{N} \lambda_{i}{\vec{Y}}^{(i)}$ where all the $x_i$ and $\lambda_i$ are unknowns, and then solving the system composed of the minors of size $t+1$ of $\widetilde{\vec{E}}$.
The Hilbert series of the ideal generated by the system is \begin{align*}
	HS(x) &=\left[ (1-x)^{(m-t)(n-t)-(K+1)}\frac{\det(A(x))}{x^{\binom{t}{2}}}\right],\\
	\text{with }    A(x) &=
	\begin{pmatrix}
		\sum_{\ell=0}^{\max(m-i,n-j)} \binom{m-i}{\ell}\binom{n-j}{\ell}x^\ell
	\end{pmatrix}_{\substack{1\le i \le t, 1\le j \le t}} 
\end{align*}
The complexity is then 
$$
\mathsf{C}_{\mathrm{Minors}}(q,m,n,k,t) = \widetilde{{O}}\left(\binom{k+D}{D}^\omega\right)
$$ 
where $D= \operatorname{deg}(HS(x))+1$ \cite{FSS10,FSS13}.
The specializations of variables must be made to lower the complexity (by optimizing over $\ell$ and $v$, as was done previously) 
$$
	{O}\Big(q^{\ell t+v}\mathsf{C}_{\mathrm{Minors}}(q, m, n-a-\ell, k -a m+ \delta(n - t+\delta) + a (t-\delta)-\ell m-v, t-\delta)\Big)
$$

\subsection{Attacking the stationary-$\MinRank$ problem}

We are now interested in solving stationary-$\minrank$. For that, we will try to adapt the attacks on $\MinRank$, and see how they perform. Essentially, one could try to adapt the kernel attack \ref{alg:kernelattack} or the support minors modelling from the previous section, using the fact that all errors have the same support. 
\medskip

{\bf \noindent Adapting the kernel attack.} Let $\vec{B}_1^{(1)},\dots,\vec{B}_k^{(1)},\vec{Y}^{(1)},\dots,\vec{B}_1^{(N)},\dots,\vec{B}_k^{(N)},\vec{Y}^{(N)} \in \Fq^{m\times n}$ be an instance of the stationary-$\MinRank$ problem. Adapting the kernel attack comes down to try and find a matrix $\vec{K} \in \Fq^{\ell \times m}$ such that 
$$
\vec{K}\left(\vec{Y}^{(j)}+\sum\limits_{i=1}^{k} x_i^{(j)}\vec{B}_i^{(j)}\right) = \vec{0}
$$ 
for all $j \in [1,\dots,N]$, where $\ell$ is as previously a value such that the linear system possesses enough equations. Note that we took the left kernel here to take advantage of the fact that the support of the errors is the same. This attack works exactly the same as the kernel attack. However, although the number of equations doubles, so does the number of variables. The result of this is that there is no gain in the running-time compared to attacking a single instance of the problem.
\medskip 

{\bf \noindent Adapting support minors.} The same goes for the support minors modeling: more equations are available to an attacker, at the price of more variables and monomials. In particular, there will be $N \cdot m \cdot \binom{n}{t+1}$ equations of degree two, but $\binom{n}{t}\cdot k \cdot N$ monomials of degree two. Note that this is also the case in higher degrees, as the increase of the degree is done by multiplying the equations by monomials composed of the $k \cdot N$ variables. Once again, this prevents any gain in the running-time, as this will not allow linearization earlier than for the usual attack. 
\medskip

{\bf \noindent Using $\MSL$.} One could also try to attack this by using $\MSL$. More precisely, for a given number of instances, say, $\ell \le N$ , a way to solve could be to build the code 
$$
\mathbf{Span}\left( \vec{B}_1^{(1)},\dots,\vec{B}_k^{(1)},\dots,\vec{B}_1^{(\ell)},\dots,\vec{B}_k^{(\ell)} \right).
$$ 
Then, the associated $\MinRank$ instances correspond to an $\MSL$ instance with parameters 
$$
(m,n,\ell,k\ell,t,q).
$$ 
This approach can obviously only be used if $k\ell$ is not greater than $mn$ otherwise the whole space is saturated and many solutions will exist. However, this is not very efficient: the dimension of the code is multiplied by $\ell$, which negates all the improvements an $\MSL$ instance could bring. It is thus never worth considering such an attack.  

As a result, the most efficient attacks on this problem are the usual attacks on $\MinRank$ previously described (we stress that in this case the specialization mentioned above does not apply, thus using a usual $\MinRank$ instance).

\subsection{Relations between the problems}

We now explain briefly the relations between different problems. First, in the same way that decoding a random $\mathbb{F}_{q^{m}}$--linear code via a $\MinRank$ solver, $\RSL$ can be solved through an $\MSL$ solver, by considering the matrix code associated to the~$\Fqm$--linear code. Then, $\MSL$ and stationary-$\MinRank$ are related as seen previously, and it is obvious that stationary-$\MinRank$ is a generalization of $\MSL$. Finally, stationary-$\MinRank$ can be solved through $\MinRank$, by considering only one of the instances. Although it does not prove its hardness, the fact that the adaptations of $\MinRank$ solvers do not work better seems to indicate that it is indeed closer to $\MinRank$ than $\MSL$.

\begin{figure}[h]
	\centering
	\scalebox{0.8}{
		\begin{tikzpicture}[->,>=stealth,thick,node distance=4cm]
\tikzstyle{node}=[circle, draw, align=center, minimum width=2.5cm]
			
\node[node] (SM) {\textsf{Stationary }\\ $\MinRank$};
			\node[node, right of=SM, xshift = 2cm] (MSL) {$\MinRank$ \\ \textsf{Support }\\\textsf{Learning}};
			\node[node, left of=SM] (M) {$\MinRank$};
			\node[node, above of=MSL] (RSL) {$\RSL$};
			\node[node, above of=M] (RSD) {Rank decoding};
\draw[->,>=stealth,dashed] (SM) edge[bend left] (MSL)
			node[xshift=3.2cm,yshift=-2.3cm,above,align=center]
			{\textsf{Generalization }\\ \textsf{of the problem}};
			\draw[->,>=stealth,thick] (MSL) edge[bend left] (SM)
			node[xshift=-2.8cm,yshift=1.3cm,above, align=center]
			{\textsf{Increase} \\ \textsf{of dimension}};
			\draw (SM) -- (M);
			\draw (RSL) -- (RSD);
			\draw (RSD) -- (M);
			\draw (RSL) -- (MSL);
			\node[
			draw,
			dashed,
			fit=(M)(SM),
			inner sep=15pt
			] (groupbox1) {};
			\node[
			draw,
			dashed,
			fit=(M)(SM)(MSL),
			inner sep=40pt
			] (groupbox) {};
			\node[below=0pt of groupbox1, align=center]
			{\textsf{Same parameters for $\MinRank$}};
			\node[below=0pt of groupbox, align=center]
			{\textsf{Use $\MinRank$ solvers}};
	\end{tikzpicture}}
	\caption{The notation $\mathsf{A} \rightarrow \mathsf{B}$ indicates that problem $\mathsf{A}$ reduces to problem $\mathsf{B}$. In other words, solving the problem $\mathsf{B}$ leads to a solution of the problem $\mathsf{A}$.  A dashed arrow $\dashrightarrow$ means there is a change of parameters in the reduction, making it impractical for a large set of parameters. Finally, what we mean by ``same parameters for $\MinRank$" is that when stationary-$\minrank$ is outside the parameters for which a reduction to $\MSL$ is possible, there is at the moment no better way than to attack only one of the $\MinRank$ instances.}
\end{figure}
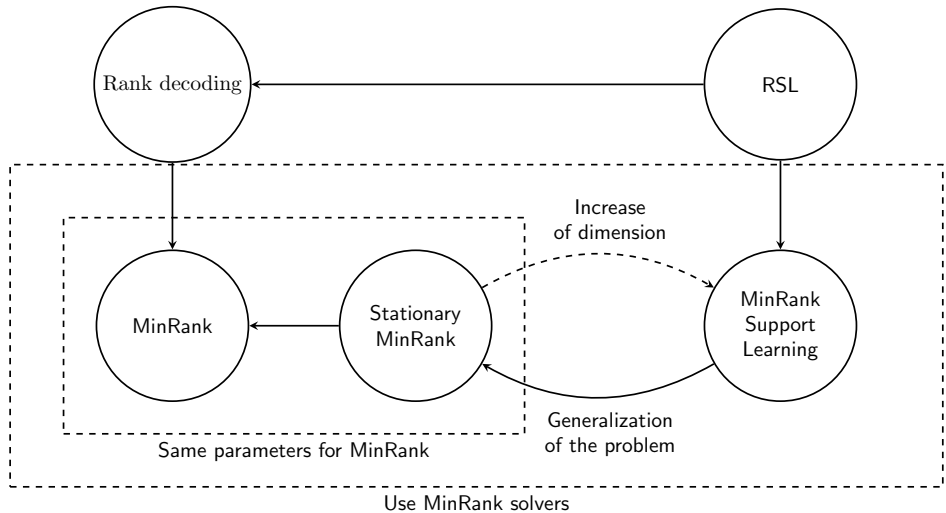
 \section{Parameters} \label{sec:parameters}
\subsection{Sizes and performances of the scheme}\mbox{ }
\medskip

{\bf \noindent Sizes of public key and ciphertext.} Thanks to the primal and dual formulations of the $\MinRank$ problem, it is always possible to reduce the sizes of the public-key and ciphertexts. 
For instance, using $\vec{EH}^\top \in \Fq^{\ell_1 \times (mn-k)}$ instead of $\vec{SG}+\vec{E}$ greatly reduces the size of the public-key, bringing it down to 
$$
\ell_1 \cdot (mn-k)\cdot \log_2(q) \ .
$$ 
For the ciphertext, it has a size given by 
$$
(\ell_1 \cdot \ell_2 + \ell_2 \cdot k)\cdot \log_2(q) \ .
$$ 
We present parameters of \textsf{MinRankPKE}, described in Section~\ref{sec:MRPKE}, in Table \ref{tab:parametres}, where the running times of attacks are the minimum between solving $\MSL$ with parameters $(q,m,n,k,r,\ell_1)$ or $(q,m,n,mn-k-\ell_1,d,\ell_2)$. Note that it is also possible to take unbalanced parameters to reduce either the size of the ciphertext or the size of the public-key.

\begin{table}[h] 
	\centering
	\scalebox{0.9}{
		\renewcommand{\arraystretch}{1.2} \begin{tabular}{|c||c|c|c|c|c|c|c|c||c|c|c||c|c|}
			\hline
			\multirow{2}{*}{Security} & \multicolumn{8}{c||}{Parameters} & \multicolumn{3}{c||}{Complexity of attacks} & \multicolumn{2}{c|}{Sizes} \\
			\cline{2-14}
			& $q$ & $m$ & $n$ & $k$ & $r$ & $d$ & $\ell_1$ & $\ell_2$ 
			& $\mathsf{C}_{\mathrm{Kernel}}$ & $\mathsf{C}_{\mathrm{SupportMinors}}$ & $\mathsf{C}_{\mathrm{Minors}}$ & $|\pk|$ & $|\mathsf{ct}|$  \\
			\hline
			\multirow{1}{*}{I} & 2 & 81 & 81 & 3201 & 4 & 4 & 35 & 35 
			&  $2^{150}$ & $2^{150}$&$2^{150}$ &  14 700 B 
			& 14 158 B\\
			\hline
			\multirow{1}{*}{III} & 2 & 103 & 103 & 5270
			& 5 & 5 & 53 & 53 
			& $2^{226}$ & $2^{207}$ & $2^{249}$ & 35 370 B 
			& 35 365 B  \\
			\hline
			\multirow{1}{*}{V} & 2 & 115 & 115 & 6613
			& 6 & 6 & 75 & 75 
			& $2^{298}$ & $2^{272}$ & $2^{325}$ & 62 020 B 
			& 62 700 B  \\
			\hline
	\end{tabular}}
	\vspace{1.5mm}
	\caption{Parameters and sizes for $\textsf{MinRankPKE}$ according to the NIST security levels I, III and V, with running time of the main known attacks, size of the public-key, and size of the ciphertext (in Bytes), taking $\omega = 2.8$. Parameters are chosen such that $|\pk| \approx |\mathsf{ct}|$. Furthermore, we always take $\kappa$, the dimension of the Gabidulin code, to be $3$.}\label{tab:parametres}
	
\end{table}

\begin{comment}
	\begin{table}[h]
		\centering
		\renewcommand{\arraystretch}{1.2} \begin{tabular}{|c||c|c|c|c|c|c|c|c||c|c|c||c|c|}
			\hline
			\multirow{2}{*}{Security level} & \multicolumn{8}{c|}{Parameters} & \multicolumn{3}{c||}{Complexity of attacks} & \multicolumn{2}{c|}{Sizes} \\
			\cline{2-14}
			& $q$ & $m$ & $n$ & $K$ & $r$ & $d$ & $\ell_1$ & $\ell_2$ 
			& $\mathsf{C}_{\mathrm{Kernel}}$ & $\mathsf{C}_{\mathrm{SupportMinors}}$ & $\mathsf{C}_{\mathrm{Minors}}$ & $|\pk|$ & $|\mathsf{ct}|$  \\
			\hline
			\multirow{1}{*}{I} & 2 & 82 & 82 & 4500 & 6 & 3 & 39 & 39 
			&  $2^{147}$ & $2^{158}$&$2^{176}$ &  21 954 B 
			& 11 033 B\\
			\hline
			\multirow{1}{*}{III} & 2 & 96 & 96 & 3292
			& 7 & 4 & 59 & 59 
			& $2^{207}$ & $2^{223}$ & $2^{249}$ & 43 714 B 
			& 24 714 B  \\
			\hline
			\multirow{1}{*}{V} & 2 & 108 & 108 & 4402
			& 8 & 5 & 83 & 83 
			& $2^{273}$ & $2^{278}$ & $2^{308}$ & 75 368 B 
			& 46 532 B  \\
			\hline
		\end{tabular}
		\vspace{1.5mm}
		\caption{Parameters and sizes for $\mathrm{MinRankPKE}$, with complexities of the main known attacks, size of the public key, and size of the ciphertext (in Bytes), taking $\omega = 2.8$. Parameters are taken to optimize either $|\pk|$ or $|\mathsf{ct}|$.} \label{tab:unbalanced}
	\end{table}
\end{comment}
\medskip 

{\bf \noindent Performances of the scheme.} In Table \ref{tab:performances}, we give the performances of the scheme for the balanced parameter sets, in millions of CPU Cycles. The implementation has been done using \texttt{rbc-lib} \cite{rbclib}, and has room for many improvements as it is a naive implementation. Still, this allows us to see that the scheme has the potential to be fast and is competitive with other encryption schemes like \textsf{FrodoKEM}. Furthermore, it has potential for a lot of parallelization as computations are mostly matrix multiplications.

\begin{table}[h] 
	\centering
	\scalebox{0.9}{
		\renewcommand{\arraystretch}{1.2}
		\begin{tabular}{|c||c|c|c|c|c|c|c|c||c|c|c|}
			\hline
			\multirow{2}{*}{Security} & \multicolumn{8}{c|}{Parameters} & \multicolumn{3}{c|}{Performances (M)} \\
			\cline{2-12}
			& $q$ & $m$ & $n$ & $k$ & $r$ & $d$ & $\ell_1$ & $\ell_2$ 
			& KeyGen & Encryption & Decryption \\
			\hline
			I & 2 & 81 & 81 & 3201 & 4 & 4 & 35 & 35 
			& 4.2 ms & 4.3 ms & 0.3 ms \\
			\hline
			III  & 2 & 103 & 103 & 5270
			& 5 & 5 & 53 & 53 
			& 11.6 ms & 11.3 ms & 0.5 ms \\
			\hline
			V & 2 & 115 & 115 & 6613
			& 6 & 6 & 75 & 75 
			& 20.0 ms & 20.1 ms & 1 ms \\
			\hline
	\end{tabular}}
	\vspace{1.5mm}
	\caption{Parameters and performances of $\mathsf{MinRankPKE}$, with timing results for key generation, encryption, and decryption (in milliseconds). The tests were run on an Intel® Core™ i7-1365U (13th Gen, 12 threads) with 32 GB RAM.}\label{tab:performances}
\end{table}

\subsection{Comparison with other encryption schemes}

We propose in Table \ref{tab:comparaison} a comparison of $\mathsf{MinRankPKE}$ with other encryption schemes. As expected, our scheme is less efficient than schemes such as \textsf{HQC} or \textsf{RQC} \cite{HQC,RQC,ABD24} but our scheme has the benefit of not having a structure like quasi-cyclicity or $\mathbb{F}_{q^{m}}$--linearity. It also remains close to \textsf{FrodoKEM} with only $5$kB difference in the public-key and in the ciphertext sizes, while relying on small fields. Furthermore, the combination of the public-key and ciphertext sizes compares well with original McEliece instantiation or other matrix code encryptions such as the scheme from \cite{ACD24}, or Loidreau's cryptosystem~\cite{Loidreau,Pham}. Compared to some schemes relying on unstructured $\RSL$ \cite{BGHO}, we perform better by a significant margin. What this table shows is that more structure (and thus more security assumptions) implies more efficient schemes. It is necessary to have both kinds of schemes, as the structures (arising from groups) may potentially turn out to be unsecure (at least adding these structures a priori decreases the security), even if all the schemes in the table are secure according to current attacks. Of course, our scheme still does not rely on plain $\MinRank$. However, as seen in the previous sections, the problem we introduced is solved only through $\MinRank$ solver and possesses a search-to-decision reduction, which are strong arguments concerning the security of the assumption.

\begin{table}[h] 
	\centering
	\renewcommand{\arraystretch}{1.2}
	\scalebox{0.7}{\begin{tabular}{|c||c|c|c|c|c|c|c|}
			\hline
			Scheme & Metric & No ideal structure & No masking of a code & No extension & No large field &$\mathsf{pk}$ & $\mathsf{ct}$ \\
			\hline
			Kyber \cite{Kyber} & Euclidean & \xmark & \cmark & \cmark & \xmark & 0.8 kB       &  0.8 kB       \\
			HQC \cite{HQC} & Hamming & \xmark & \cmark & \cmark & \cmark &2.2  kB    &    4.4 kB   \\
			RQC \cite{RQC} & Rank & \xmark & \cmark & \xmark & \cmark & 0.3 kB       &  1.1 kB       \\
			LowMS \cite{ADG+} & Rank & \cmark & \xmark & \xmark & \cmark & 4.77 kB & 1.14 kB\\
			Loidreau \cite[Conclusion]{Pham} & Rank & \cmark & \xmark & \xmark & \cmark & 34.5 kB & 1.8 kB \\ 
			Generalization of Loidreau \cite{EL25} & Rank & \cmark & \xmark & \xmark & \cmark  & 9.5 kB & 0.94 kB \\
			McEliece \cite{McEliece} & Hamming & \cmark & \xmark & \cmark & \cmark & 261 kB        &  96 B      \\
			MinRank Gabidulin \cite{ACD24}& Rank & \cmark & \xmark & \cmark & \cmark  & 33-78 kB        & 207-84 B    \\
			FrodoKEM \cite{Frodo} & Euclidean & \cmark & \cmark & \cmark & \xmark &9.6 kB  & 9.7 kB   \\
			Multi-UR-AG \cite{BBBG24} & Rank & \cmark & \cmark & \xmark & \cmark & 4.1 kB       &  6.9 kB       \\
			Injective Rank Trapdoor \cite{BGHO}& Rank & \cmark & \cmark & \xmark & \cmark & 203 kB & 1663 kB \\
			Alekhnovich \cite{A03} & Hamming & \cmark & \cmark & \cmark & \cmark & $\ge$ MB*       & $\ge$ MB*           \\
			\textbf{\textrm{MinRankPKE-I}}  & Rank & \cmark & \cmark & \cmark & \cmark  & \textbf{14.7 kB} & \textbf{14.1 kB} \\
			\hline
		\end{tabular}}
	\vspace{1.5mm}
	\caption{Public-key ($\mathsf{pk}$) and ciphertext sizes ($\mathsf{ct}$) of $\mathsf{MinRankPKE}$ for the security level I of the NIST. No actual instantiation of Alekhnovich's scheme has been proposed in the literature. However, we estimate the sizes as several megabytes.}
	\label{tab:comparaison}
\end{table}
 \section{Conclusion and further work}

This work presents an encryption scheme that follows Alekhnovich and Regev' framework, adapted to the rank metric using $\Fq$--linear matrix codes. Our scheme possesses several features, namely: $(i)$ a search-to-decision security reduction, $(ii)$ practical parameters and $(iii)$ having implementation performances comparable to other unstructured schemes such as \textsf{FrodoKEM}.
To study the security of our scheme, we had to introduce two new problems: stationary-$\MinRank$ and \textsf{MSL} which are closely related. We thoroughly studied their algorithmic hardness via usual $\MinRank$ solvers. Overall, this results in a practical encryption scheme, that we called \textsf{MinRankPKE}, while still keeping security reductions. We stress that this security reduction is present in very few schemes, for instance \textsf{HQC} does not benefit of such a property.

For future work, it would be interesting to instantiate our scheme by introducing some structures, in the way \textsf{HQC}, \textsf{RQC} or \textsf{Kyber} do, or on the contrary to find a way to be even closer to $\MinRank$.  \label{sec:ccl}
    \section*{Acknowledgments.}
The authors are supported by the French \emph{Agence Nationale
	de la Recherche} (ANR) through the \emph{Plan France 2030 programme}
ANR-22-PETQ-0008 ``{PQ-TLS}''. The work of Thomas Debris-Alazard was
funded through the French ANR project \emph{Jeunes Chercheuses, Jeunes
	Chercheurs} ANR-21-CE39-0011 ``{COLA}''.

\iftoggle{llncs}{ \bibliographystyle{src/splncs04}}{
\bibliographystyle{alpha}
}
\newcommand{\etalchar}[1]{$^{#1}$}

\newpage
\appendix

\section{Proof of Theorem \ref{theo:fund}}

{\bf \noindent Notation.} A {\em distinguisher} between two distributions $\mathcal{D}_{0}$ and
$\mathcal{D}_{1}$ is a probabilistic polynomial time algorithm
$\mathcal{A}$ that takes as input an oracle $\mathcal{O}_{b}$ corresponding
to a distribution $\mathcal{D}_{b}$ with $b\in \{0, 1\}$ and outputs an
element $\mathcal{A}(\mathcal{O}_{b})\in\{0, 1\}$.
Consider the following approach for solving a decision problem
between two distributions $\mathcal{D}_{0}$ and $\mathcal{D}_{1}$, pick
$b\Unif \{0, 1\}$ and answer $b$ regardless of the input. This
algorithm solves this problem with probability $1/2$ which is not
interesting. The efficiency of an algorithm $\mathcal{A}$ solving a
decision problem is measured by the difference between its probability
of success and $1/2$. The relevant quantity to consider is the
{\em advantage} defined as:
$$
Adv_{\mathcal{A}}(\mathcal{D}_{0}, \mathcal{D}_{1}) \eqdef
\dfrac{1}{2} \left( \mathbb{P}(\mathcal{A}(\mathcal{O}_{b}) =
1 \mid b = 1) - \mathbb{P}(\mathcal{A}(\mathcal{O}_{b}) = 1 \mid b = 0) \right)
$$
where the probabilities are computed over the internal randomness of
$\mathcal{A}$, a uniform $b \in \{0,1\}$ and inputs according to a distribution
$\mathcal{D}_{b}$. The advantage of a distinguisher
$\mathcal{A}$ measures how good it is to solve a distinguishing problem.
Indeed, it is classical fact that:
$$
\mathbb{P}(\mathcal{A}(\mathcal{O}_{b}) = b) = \frac{1}{2} + Adv_{\mathcal{A}}(\mathcal{D}_{0}, \mathcal{D}_{1}).
$$
\begin{remark}
	Even if it means answering $1-\mathcal{A}(\mathcal{O}_{b})$ instead of
	$\mathcal{A}(\mathcal{O}_{b})$, the advantage can always be assumed to be
	a positive quantity.
\end{remark}

Let us start by introducing the decisional version of stationary-$\minrank$.

		\begin{definition}[decisional stationary-$\minrank$]\label{prob:decStatMinrank} 
			Let $m,n, N, k_1,\dots,k_N,t,q$ be integers which are functions of some security
			parameter $\lambda$ and such that $mn \geq k_j$ for all $j \in [1,N]$. Let $(\vec{B}_{\ell}^{j})_{j \in [1,N]}, \vec{Y}_{0}^{(j)}$ for $j \in [1,N]$ be sampled as in stationary-$\minrank(m,n,N,(k_i)_{i \in [1,N]}, t,q)$ and $\vec{Y}_{1}^{(j)} \in \mathbb{F}_{q}^{m \times n}$ for~$j \in [1,N]$ be sampled uniformly at random.

			Let $b \in \{0,1\}$ be a uniform bit. The decisional stationary-$\minrank(m,n,N,(k_i)_{i \in [1,N]},t,q)$ problem consists, given $\left( \left( \vec{B}_{\ell}^{(j)} \right)_{\ell \in [1,k_j]}, \vec{Y}_{b}^{(j)} \right)_{j \in [1,N]}$, in finding~$b$. 
		\end{definition}

		\begin{definition}[stationary-$\minrank$ advantage]
			Let $\mathcal{X}_{b} \eqdef \left( \left(  \vec{B}_{\ell}^{(j)} \right)_{\ell}, \vec{Y}_{b}^{(j)} \right)_{j \in [1,N]}$ and $b$ be distributed as in decisional stationary-$\minrank(m,n,N,(k_i)_{i \in [1,N]},t,q)$. 
			The stationary-$\minrank$ advantage for parameter $(m,n,N,(k_i)_{i \in [1,N]},t,q)$ in time $T$ is defined as 
			$$
			Adv^{\mbox{st-}\mathsf{Mr}}(m,n,N,(k_i)_{i \in [1,N]},t,q,T) \eqdef 
			\frac{1}{2}\max\limits_{\mathcal{A}} Adv_{\mathcal{A}}(\mathcal{X}_{0}, \mathcal{X}_{1})	
$$
			where the maximum is taken over all the algorithms $\mathcal{A}$ running in time $\leq T$.
\end{definition}

		In the following lemma we show that breaking our scheme implies an algorithm to solve decisional stationary-$\minrank$. 
		\begin{lemma}\label{lemma:advScheme}
			Consider an attacker against the public-key encryption scheme described in Figure~\ref{algo:scheme}. Suppose that this attacker extracts an encrypted bit in time $T$ with probability $1/2+\varepsilon$ by using the knowledge of the public-key. Then,
			$$
			\varepsilon_{0} \geq \frac{\varepsilon}{2} \quad \mbox{or} \quad \varepsilon_1 \geq \frac{\varepsilon}{2} \quad \mbox{ where } \; \left\{ \begin{array}{l} \varepsilon_0 \eqdef Adv^{\mbox{st-}\mathsf{Mr}}(m,n,\ell_1,(k_j)_{j \in [1,\ell_1]},r,q,T) \\
				\varepsilon_1 \eqdef Adv^{\mbox{st-}\mathsf{Mr}}(m,n,N,(mn/\ell_2-k'_i)_{i\in [1,\ell_2]},d,q,T)  \end{array}\right. \ .
			$$
			where $\sum_{j=1}^{\ell_1} k_j +\ell_1 = \sum_{i=1}^{\ell_2} k_i'$. 
\end{lemma}
		\begin{proof}
			Suppose that we replace public-keys in our scheme by perfectly random matrices. Let~$1/2+\varepsilon'$ be the probability of that attacker to succeed to break this version of the scheme. We clearly have 
			\begin{equation}\label{eq:epsPrime} 
			|\varepsilon'-\varepsilon| \leq \varepsilon_{0} \Longrightarrow \varepsilon' \geq \varepsilon - \varepsilon_{0}
		\end{equation}

			Suppose now that we are given an instance of stationary-$\minrank$ for parameters $(m,n,\ell_2,(mn/\ell_2-k'_i)_{i \in [1,\ell_2]},d,q)$:
			$$
			\left( \left( \vec{B}^{(j)}_{\ell}\right)_{\ell \in [1,mn/\ell_2-k'_j]},\vec{Y}^{(j)}_{b}\right)_{j \in [1,\ell_2]}
			$$
			Let us consider the code 
			$$
			\mathcal{D} \eqdef \mathbf{Span}\left( \left( \left( \vec{B}^{(j)}_{\ell}\right)_{\ell \in [1,mn/\ell_2-k'_j]} \right)\right)_{j \in [1,\ell_2]}
			$$
			It has dimension $mn - \sum_{j=1}^{\ell_2} k'_j = mn - \sum_{i=1}^{\ell_1} k_i-\ell_1$. 
			Notice that $\mathcal{D}^{\perp}$ has dimension $\sum_{i=1}^{\ell_1} k_i+\ell_1$ and it is a random code. We can decompose this code as $\ell_1$ random matrix codes with dimension $k_{i}$'s with additional $\ell_1$ uniform matrices. In other words, what we have just built is just the public-key of the scheme that our considered attacker can break with probability $1/2+\varepsilon$'. Notice now that during encryption of $b = 0$, we have to sample noisy codewords with underlying code $\mathcal{D}$. Let us pick uniformly at random $\vec{D}_1,\dots,\vec{D}_{\ell_2} \Unif \mathcal{D}$ and compute 
			$$
			\forall j \in [1,\ell_2], \; \vec{Z}^{(j)}_{b} \eqdef \vec{Y}^{(j)}_b + \vec{D}_{j}
			$$
			If the $\vec{Y}^{(j)}$'s are uniformly distributed, then $\vec{Z}^{(j)}_{b}$'s are also uniformly distributed and they correspond to the encryption of $b=1$. On the other hand the $\vec{Z}^{(j)}_b$'s are distributed as the encryption of~$b= 0$. We deduce that with advantage $\varepsilon' \leq \varepsilon_1$ our attacker solves the given stationary-$\minrank$ instance. Therefore, using Equation~\eqref{eq:epsPrime}, we deduce that $\varepsilon \leq \varepsilon_0 + \varepsilon_1$ which concludes the proof. 
		\end{proof}

		The above lemma shows that to prove Theorem~\ref{theo:fund} we just have to show how from an algorithm solving  the decisional form of stationary-$\minrank$ with probability $\geq \varepsilon/2$ we deduce an algorithm solving its search counter-part. To obtain such reduction we will use Goldreich-Levin Theorem~\cite{GL89,G01_a} that we recall now.

		\begin{theorem}[Goldreich-Levin Theorem]\label{propo:Goldreich} Let $f : \mathbb{F}_{2}^{*} \rightarrow \mathbb{F}_2^{*}$, $\mathcal{A}$ be a probabilistic algorithm running in time $T(n)$ and $\varepsilon(n) \in ( 0,1 )$ be such that
			$$	
			\mathbb{P}\left( \mathcal{A}( f(\vec{x}_n),\vec{r}_{n})  = \vec{x}_{n} \cdot \vec{r}_{n} \right)  = \frac{1}{2} + \varepsilon(n)	
			$$	
			where the probability is computed over the internal coins of $\mathcal{A}$, $\vec{x}_{n}$ and $\vec{r}_{n}$ that are uniformly distributed over $\mathbb{F}_2^{n}$. 		
			Let $\ell(n) \eqdef \log(1/\varepsilon(n))$. Then, it exists an algorithm $\mathcal{A}'$ running in time~$O\left(n^{2}\ell(n)^{3}T(n)\right)$ that satisfies
			$$
			\mathbb{P}\left( \mathcal{A}'(f(\vec{x}_n)=\vec{x}_n) \right) = \Omega\left(\varepsilon(n)^{2}\right)
			$$
			where the probability is computed over the internal coins of $\mathcal{A}'$ and $\vec{x}_n$.
		\end{theorem}

		To apply this theorem in our case we will first use the following {\em hybrid argument.} Let $i \in [1,N]$ and $\mathcal{H}_{i}$ be the following distribution. We sample 
		$$
		\left( \left( \vec{B}_{\ell}^{(j)} \right)_{\ell \in [1,k_j]}, \vec{Y}^{(j)} \right)_{j \in [1,i]}
		$$ 
		as a proper stationary-$\minrank$ distribution (notice that here we have $i$ samples from a stationary-$\minrank$ instance, not $N$ samples) and we samples 
		$$
		\left( \left( \vec{B}_{\ell}^{(j)} \right)_{\ell \in [1,k_j]}, \vec{U}^{(j)} \right)_{j \in [i+1,N]}
		$$ 
		where the $\vec{B}_{\ell}^{(j)}$'s and~$\vec{U}^{(j)}$'s are uniform matrices. In particular the decisional stationary-$\minrank$ requires to distinguish between $\mathcal{H}_{0}$ and $\mathcal{H}_{N}$.

		\begin{lemma}[Hybrid argument]\label{lemma:hybrid} 
There exists $i_0\in [1,N]$ such that,
$$
			Adv_{\mathcal{A}}(\mathcal{H}_{i_0},\mathcal{H}_{i_0+1}) \geq \frac{Adv_{\mathcal{A}}(\mathcal{H}_0, \mathcal{H}_N)}{N}
$$ 
		\end{lemma}

		\begin{proof}
			The following equality holds:
			$$
			Adv_{\mathcal{A}}(\mathcal{H}_{0}, \mathcal{H}_{N}) =
			\sum_{i=0}^{N-1}\adv_{\mathcal{A}}(\mathcal{H}_{i}, \mathcal{H}_{i+1}).
			$$
			Therefore, it exists $i_{0}\in [1, N]$ such
			that
			$Adv_{\mathcal{A}}(\mathcal{H}_{i_{0}},\mathcal{H}_{i_{0}+1}) \ge \frac{Adv_{\mathcal{A}}(\mathcal{H}_{0}, \mathcal{H}_{N})}{N}$. 
		\end{proof}

		We are now ready to prove the following search-to-decision reduction. 
		\begin{theorem}[stationary-$\minrank$ search-to-decision reduction]\label{theo:searchToDecision}
			Let $\mathcal{A}$ be a probabilistic algorithm running in time $T$ whose stationary-$\minrank$ advantage is given by $\varepsilon$ for parameters $(m,n,N,(k_i)_{i \in [1,N]},t,2)$. Let~$\ell \eqdef \log(1/\varepsilon)$. 
			Then it exists an algorithm $\mathcal{A}'$ that solves stationary-$\minrank$ for parameters $(m,n,N,(k_i)_{i \in [1,N]},t,2)$ in time $O(Nmn^{2}\ell^{3})T$ and with probability $\Omega(\frac{\varepsilon^{2}}{N^{2}})$.  
		\end{theorem}

		\begin{proof}
			First, notice that a stationary-$\minrank$ instance in dual representation for $q=2$ can be written as (via the vectorization of matrices as defined in Equation~\eqref{eq:rho})
			$$
			\left( \vec{H}^{(j)}, \vec{H}^{(j)}\left( \vec{e}^{(j)}\right)^{\top} \right)_{j \in [1,N]}
			$$
			where $\vec{e}^{(j)} = \rho(\vec{E}^{(j)}) \in \mathbb{F}_{2}^{mn}$ and the $\ell$-thm row of $\vec{H}^{(j)} \in \mathbb{F}_{2}^{mn-k_j}$ is given by $\rho\left( \vec{B}_{\ell}^{(j)} \right)$ which is a uniformly distributed vector.

			In particular, using notation of Lemma~\ref{lemma:hybrid}, an instance of $\mathcal{H}_{i_0}$ can be written as 
			\begin{equation}\label{eq:fctHi0}
				\left( \vec{H}^{(j)}, \vec{H}^{(j)}\left( \vec{e}^{(j)}\right)^{\top} \right)_{j \in [1,i_0]}, \; 
				\left( \vec{H}^{(j)}, \vec{u}^{(j)} \right)_{j \in [i_0+1,N]}
			\end{equation} 
			where the $\vec{u}^{(j)}\in \mathbb{F}_{2}^{mn-k}$ are uniformly distributed. Notice that all the vectors $\vec{e}^{(j)}$ are obtained via the $\vec{E}^{(j)}$'s which have a same column support. We can therefore interpret elements of Equation~\eqref{eq:fctHi0} as the output of some function $f(\vec{e}^{(i_0)})$. The different vectors $\vec{e}^{(j)}$ are then obtained via a pseudo-random generator taking $\vec{e}_{i_0}$ as input. Our goal now is to show how from $f(\vec{e}_{i_0})$,~$\vec{r} \in \mathbb{F}_{2}^{mn}$ and a distinguisher $\mathcal{A}_{i_0}$ between $\mathcal{H}_{i_0}$ and $\mathcal{H}_{i_0+1}$ with advantage $\varepsilon'$ we can deduce $\vec{e}_{i_0} \cdot \vec{r}$ with probability~$1/2 + \varepsilon'$. 
			\medskip

			{\bf Algorithm} $\mathcal{A}'$ :

			\quad \texttt{Input}: $\left( \vec{H}^{(j)}, \vec{H}^{(j)}\left( \vec{e}^{(j)}\right)^{\top} \right)_{j \in [1,i_0]}$, 
			$\left( \vec{H}^{(j)}, \vec{u}^{(j)} \right)_{j \in [1,N]}$ and $\vec{r} \in \mathbb{F}_2^{n}$,

			\qquad 1. $\vec{u} \in \mathbb{F}_2^{mn-k_{i_0}}$ be uniformly distributed

			\qquad 2. $\vec{M}^{(i_0)} \eqdef \vec{H}^{(i_0)} - {\vec{u}}^{\top}\vec{r}$

			\qquad 3. $b$ be the output of $\mathcal{A}_{i_0}$ when we feed as input $\left( \vec{H}^{(j)}, \vec{H}^{(j)}\left( \vec{e}^{(j)}\right)^{\top} \right)_{j \in [1,i_0]}$,
			$\left( \vec{H}^{(j)}, \vec{u}^{(j)} \right)_{j \in [1,N]}$ but where we replaced $\vec{H}^{(i_0)}$ by $\vec{M}^{(i_0)}$.

			\quad \texttt{Output}: $b$
			\medskip

			The matrix $\vec{H}^{(i_0)}$ is uniformly distributed by definition, therefore $\vec{M}^{(i)}$ is also uniformly distributed. Notice now that ,
			\begin{equation*} 
				\vec{H}^{(i_0)}\left( \vec{e}^{(i_0)} \right)^{\top} = \vec{M}_{i_0}\left( \vec{e}^{(i_0)} \right)^{\top} + \left( \vec{e}^{(i_0)}\cdot\vec{r} \right)\vec{u}.
			\end{equation*}
			Let,
			\begin{equation*}
				\vec{s}' \eqdef \vec{M}_{i_0}\left( \vec{e}^{(i_0)} \right)^{\top} + \vec{u}.
			\end{equation*}
			It is readily verified that $\vec{s}'$ is uniformly distributed. Therefore, according to $ b = \vec{e}^{i_0}\cdot\vec{r} = 0$ or~$1$, we obtain distributions $\mathcal{H}_{i_0}$ or $\mathcal{H}_{i_0+1}$. Therefore, by conditioning on $\vec{e}^{(i_0)} \cdot \vec{r}$,  the probability that~$\mathcal{A}'$ outputs $\vec{x}\cdot \vec{r}$ is given by~$1/2 + \varepsilon'$.

			To conclude the proof notice that we don't have an access to a distinguisher $\mathcal{A}_{i_0}$, we only have an access to distinguisher $\mathcal{A}$ between $\mathcal{H}_{0}$ and $\mathcal{H}_{N}$ with advantage $\varepsilon$ by assumption. But by Lemma~\ref{lemma:hybrid} we can use $\mathcal{A}$ to distinguish $\mathcal{H}_{i_{0}}$ and $\mathcal{H}_{i_0+1}$ with probability $\geq \varepsilon/N$ for some unknown~$i_{0}$. What we are going to do is to use $\mathcal{A}$ for all $i \in [1,N]$ and applying the previous process $\mathcal{A}$' and then the transformation from Goldreich-Levin theorem. It concludes the proof.  
		\end{proof}
		
		All the ingredients are now in place to prove Theorem~\ref{theo:secu}.

		\begin{proof}[Proof of Theorem~\ref{theo:secu}] We simply combine Theorem~\ref{theo:searchToDecision} with Lemma~\ref{lemma:advScheme}.  
	\end{proof} \label{app:reduction}

\end{document}